\PassOptionsToPackage{unicode}{hyperref}
\PassOptionsToPackage{hyphens}{url}
\documentclass[
]{article}
\usepackage{amsmath,amssymb}
\usepackage{amsthm}
\usepackage{lmodern}
\usepackage{ifxetex,ifluatex}
\ifnum 0\ifxetex 1\fi\ifluatex 1\fi=0 
  \usepackage[T1]{fontenc}
  \usepackage[utf8]{inputenc}
  \usepackage{textcomp} 
\else 
  \usepackage{unicode-math}
  \defaultfontfeatures{Scale=MatchLowercase}
  \defaultfontfeatures[\rmfamily]{Ligatures=TeX,Scale=1}
\fi
\IfFileExists{upquote.sty}{\usepackage{upquote}}{}
\IfFileExists{microtype.sty}{
  \usepackage[]{microtype}
  \UseMicrotypeSet[protrusion]{basicmath} 
}{}
\usepackage{xcolor}
\IfFileExists{xurl.sty}{\usepackage{xurl}}{} 
\IfFileExists{bookmark.sty}{\usepackage{bookmark}}{\usepackage{hyperref}}
\hypersetup{
  pdftitle={Improvements to Modern Portfolio Theory based models applied to electricity systems.},
  pdfauthor={Gabriel Malta Castro, Claude Klöckl, Peter Regner, Johannes Schmidt, Amaro Olimpio Pereira Jr.},
  hidelinks,
  pdfcreator={LaTeX via pandoc}}
\usepackage[margin=1in]{geometry}
\usepackage{longtable,booktabs,array}
\usepackage{calc} 
\usepackage{etoolbox}

\usepackage[affil-it]{authblk}

\makeatletter
\patchcmd\longtable{\par}{\if@noskipsec\mbox{}\fi\par}{}{}
\makeatother
\IfFileExists{footnotehyper.sty}{\usepackage{footnotehyper}}{\usepackage{footnote}}
\makesavenoteenv{longtable}
\usepackage{graphicx}
\makeatletter
\def\maxwidth{\ifdim\Gin@nat@width>\linewidth\linewidth\else\Gin@nat@width\fi}
\def\maxheight{\ifdim\Gin@nat@height>\textheight\textheight\else\Gin@nat@height\fi}
\makeatother
\setkeys{Gin}{width=\maxwidth,height=\maxheight,keepaspectratio}
\makeatletter
\def\fps@figure{htbp}
\makeatother
\usepackage{indentfirst}
\ifluatex
  \usepackage{selnolig}  
\fi
\newlength{\cslhangindent}
\newlength{\csllabelwidth}
 {
  \setlength{\parindent}{0pt}
  \ifodd #1 \everypar{\setlength{\hangindent}{\cslhangindent}}\ignorespaces\fi
  \ifnum #2 > 0
  \setlength{\parskip}{#2\baselineskip}
  \fi
 }%
 {}
\usepackage{calc}
\usepackage[backend=biber,style=authoryear,maxcitenames=2,maxbibnames=4,sorting=nty,mergedate=basic]{biblatex}

\DeclareNameAlias{author}{family-given}
\newtheorem{theorem}{Theorem}
\newtheorem{definition}{Definition}

\title{Improvements to Modern Portfolio Theory based models applied to electricity systems}

\author[1]{ Gabriel Malta Castro}
\author[2] {Claude Klöckl\thanks{Corresponding author: Feistmantelstraße 4, 1180 Wien, Austria. claude.kloeckl@boku.ac.at}}
\author[2] {Peter Regner}
\author[2] {Johannes Schmidt}
\author[1] {Amaro Olimpio Pereira Jr.}

\affil[1]{%
  Energy Planning Program, Graduate School of Engineering, Federal University of Rio de Janeiro, Rio de Janeiro, Brazil.}
\affil[2]{%
  Institute for Sustainable Economic Development, University of Natural Resources and Life Science, Vienna, Austria}

\begin{document}
\maketitle
\begin{abstract}
With the increase of variable renewable energy sources (VRES) share in electricity systems, many studies were developed in order to determine their optimal technological and spatial mix.
Modern Portfolio Theory (MPT) has been frequently applied in this context.
However, some crucial aspects, important in energy planning, are not addressed by these analyses.
We, therefore, propose several improvements and evaluate how each change in formulation impacts results. More specifically, we address generation costs, system demand, and firm energy output, present a formal model and apply it to the case of Brazil. 
We found that, after including our proposed modifications, the resulting efficient frontier differs strongly from the one obtained in the original formulation. Portfolios with high output standard deviation are not able to provide a firm output level at competitive costs.
Furthermore, we show that diversification plays an important role in smoothing output from VRES portfolios.
\end{abstract}

\hypertarget{keywords}{%
\section*{Keywords}\label{keywords}}
\addcontentsline{toc}{section}{keywords}

Optimization, Diversification, Portfolio selection, Renewable energy sources, CVaR.

\hypertarget{definitions}{%
\section*{Definitions}\label{definitions}}
\addcontentsline{toc}{section}{Definitions}

\begin{description}
\item[CF]
Capacity factor
\item[CV]
Coefficient of variation
\item[CVaR]
Conditional Value-at-Risk
\item[LCOE]
Levelized Cost of Electricity
\item[MPT]
Modern Portfolio Theory
\item[PV]
Photovoltaic
\item[SD]
Standard deviation
\item[VaR]
Value-at-Risk
\item[VRES]
Variable renewable energy sources
\end{description}

\hypertarget{introduction}{%
\section{Introduction}\label{introduction}}

Variable renewable energy sources (VRES), such as wind and photovoltaic (PV) power plants, have increased their share in electricity systems all over the world during this century.
In the next decades, the share of those sources in global power systems tends to increase even more, possibly reaching (almost) 100\% of renewable source penetration, as, following Paris Agreement consensus, the world will have to reduce its greenhouse gas emissions to avoid a more severe climate change \parencite{bogdanovRadicalTransformationPathway2019, delucchiProvidingAllGlobal2011, haegelTerawattscalePhotovoltaicsTransform2019, jacobsonMatchingDemandSupply2018a, schmidtOptimalMixSolar2016}.
As VRES cannot be dispatched, a system with high shares of VRES may run the risk of not being able to meet demand in some moments.
Furthermore, it is possible that total generation is sometimes higher than load.
To mitigate those imbalances, storage systems and demand response schemes will play an important role in balancing supply and demand  \parencite{delucchiProvidingAllGlobal2011}.

However, the need for such integration technologies can be reduced if variability in generation is reduced, potentially resulting in a higher level of firm power supply. A way to reduce variability is by increasing technological and spatial diversification, i.e., by combining generation from different technologies at different locations. 
Numerous studies have shown that a system consisting of power plants that have low or negative correlation among themselves can smooth out the resulting generation profile by lowering its variance.
Many studies have investigated the complementarity of renewable energy sources in different locations to better understand and quantify it\footnote{A recent review on complementarity of VRES has been published by \textcite{juraszReviewComplementarityRenewable2020}}.

Some authors have used optimization models to determine the optimum mix of locations and technologies. Often, the main methodology used is based on Markowitz Modern Portfolio Theory (MPT) \parencite{markowitzPortfolioSelection1952,dellano-pazEnergyPlanningModern2017}.
Originally, MPT is applied to find the optimal portfolios of financial assets, based on three parameters:

\begin{enumerate}
\def\labelenumi{\arabic{enumi}.}
\item
  maximize return or yield, given by average return from assets that compose the portfolio;
\item
  minimize risk, given by standard deviation (SD) of the portfolio return in a given period;
\item
  maintain a fixed budget.
\end{enumerate}

In general, when MPT formulation has been applied to VRES systems, the portfolio capacity factor is the ``yield'' or ``return,'' generation standard deviation is the ``risk'' and portfolio installed capacity is the ``budget'' \parencite{chuppOptimalWindPortfolios2012,cunhaDesigningElectricityGeneration2015,degeilhQuantitativeApproachWind2011,drakeWhatExpectGreater2007,huGeographicalOptimizationVariable2019,rombautsOptimalPortfoliotheorybasedAllocation2011,roquesOptimalWindPower2010,santos-alamillosExploringMeanvariancePortfolio2017,scalaPortfolioAnalysisGeographical2019,shahriariCapacityValueOptimal2018,thomaidisOptimalManagementWind2016a}.
The set of optimal solutions is used to construct the \emph{efficient frontier}. A point is part of the frontier if it is not dominated by any other feasible portfolio, i.e. if there is no portfolio that is better in both expected generation and variance whilst maintaining the installed capacity at a fixed value.

Some innovations and changes in the formulation to address issues related to VRES characteristics have been proposed.
For instance, \textcite{roquesOptimalWindPower2010} restrict output at peak-load hours as ``yield'', others use, as ``risk'', the standard deviation of hourly output differences instead of hourly output \parencite{novacheckDiversifyingWindPower2017, rombautsOptimalPortfoliotheorybasedAllocation2011,roquesOptimalWindPower2010}, \textcite{rombautsOptimalPortfoliotheorybasedAllocation2011} represent transmission capacity constraints, and \textcite{shahriariCapacityValueOptimal2018} analyze different scales of spatial and temporal aggregation.
However, with the exception of \textcite{novacheckDiversifyingWindPower2017}\footnote{In this work, the authors minimize portfolio installed capacity for a given mean output level.}, all studies maintain the same structure of ``return,'' ``risk'' and ``budget,'' i.e., mean output as ``return,'' output standard deviation as ``risk'' and installed capacity as ``budget.''
We propose, in this work, some changes in this structure in order to better reflect the goals usually optimized for electricity systems: reliability and low cost.

First, we propose to change the measure of risk.
Standard deviation metrics measure deviations from the mean symmetrically. Therefore, output deviations, regardless of being higher or lower than the expected level, have the same impact in the MPT model.
That symmetry is poorly justified as energy shortfall usually has higher impacts than energy surplus.
If portfolio generation is below expected levels, backup power plants have to be used to supply the demand, incurring additional costs.
In some extreme cases, no plant will be available as a backup and part of the load will not be met, causing a loss-of-load event and consequently very significant cost.
On the other hand, excess generation may require some management to avoid instabilities in the grid, but excess electricity can be stored for later use, can supply controllable secondary loads or can be simply curtailed \parencite{haleIntegratingSolarFlorida2018, nelsonInvestigatingEconomicValue2018}. Costs of excess generation are therefore very low or even negative.

Moreover, the variance by itself does not fully determine the shape of the distribution of the portfolio generation.
\textcite{huGeographicalOptimizationVariable2019} have shown, in an ex-post analysis after performing an efficient frontier optimization, that the regions in the extremes of the efficient frontier curve (portfolios with high and low standard deviation) tend to have lower values for the power output at lower percentiles. Therefore, some points in the curve may be providing more firm capacity to the system than others.
Therefore, we propose, as an improvement to the formulation, to incorporate a constraint to obtain portfolios that are able to maintain a minimum generation level at a given risk.

The studies that use MPT optimization have, as one objective, the minimization of the SD of portfolio generation.
However, this ignores load variability and its interaction with portfolio generation.
Therefore, instead of minimizing the standard deviation of generation, we propose to minimize the standard deviation of electricity balance\footnote{Electricity balance is the difference between portfolio generation and system demand.}.
\textcite{degeilhQuantitativeApproachWind2011} propose a method to include system demand in the formulation. They concluded that its effect would be negligible, and thus the authors did not further explore it. However, they assumed that the load would have the same weight in the system as one single wind turbine and, in their data set, the correlation of demand to each power plant output was close to 0. The latter reason is specific to their data set and the former can be solved by defining a proper weight to load in comparison to the portfolio. In the present work, we propose a way to incorporate load in the optimization problem as a power plant which production is negative.

As the ``return'' component, instead of using the capacity factor, we propose to use power plant levelized cost of electricity (LCOE), as this better reflects system planning goals. This change includes plant costs and maintains capacity factor of each plant as a decision factor, as it is a component of LCOE calculation.

To test and compare those different improvements, a case study was performed using wind power and PV generation as well as demand data from Brazil.
Due to its large territorial extension, Brazil is an appropriate case for our study, as geographical dispersion is more prone to less correlated generation profiles. Each change in the formulation will be introduced at a time, in order to assess its individual impact on the results.

The remainder of this text is organized as follows.
Section \ref{methods} presents the methodology and how each improvement changes the model formulation. Section \ref{results-and-discussion} compares results of the model with and without the proposed changes for each change individually and discusses them. Finally, we present general implications of our improvements in Section \ref{conclusions}.

\hypertarget{methods}{%
\section{Methods and Data}\label{methods}}
The MPT formulation in its original form, applied to financial assets, is outlined in the following. Let \(r_{i}\) be a random variable giving the return of asset \(i\),  \(R_i = E(r_i)\) is its expected value and \(\sigma_{ij} = cov(r_i,r_j)\) is the covariance of returns from assets \(i\) and \(j\). Therefore, considering no short position is allowed:

\begin{align}
\label{eq:MarkObj}
& \underset{X_{i}}{\text{maximize}} & & R(X_{i}) = \sum_{i=1}^{N} X_i R_i \\
\label{eq:MarkSD}
& \text{subject to} & & \sqrt{\sum_{i=1}^{N}\sum_{j=1}^{N}X_iX_j\sigma_{ij}} \le \sigma_P  \\
\label{eq:MarkSumX}
&&& \sum_{i=1}^{N}X_i = 1 &\\
 \label{eq:MarkNonNeg}
& & & X_{i} \ge 0 \quad \forall i & 
\end{align}

\(R(X_i)\) is portfolio return, \(X_i\) is the share of asset \(i\) in the portfolio, \(N\) is the number of assets and \(\sigma_P\) is the maximum allowed portfolio standard deviation.
The optimization is solved for different values of \(\sigma_P\in [0,\max_{i}(\sigma_{i})]\) in order to obtain the efficient frontier\footnote{There are two equivalent alternative formulations. In the first one, portfolio variance is minimized and portfolio output is constrained and its minimum value varies at each iteration. In the second one, constraint \eqref{eq:MarkSD} is removed and the objective function changes to \(\sum_{i=1}^{N} X_i R_i - \lambda \sqrt{\sum_{i=1}^{N}\sum_{j=1}^{N}X_iX_j\sigma_{ij}}\). In this case, the problem is solved for different values of \(\lambda \in [0,\infty]\) in order to obtain the efficient frontier.}.
As stated in equation \eqref{eq:MarkSumX}, the sum of weights is constrained to 1. Consequently, the solution gives the share of each asset in the portfolio instead of the absolute amount of money invested in each asset.

When adapted to energy systems, the mathematical formulation is very similar.
Instead of financial assets, the portfolio consists of VRES plants with different output profiles.
\(X_i\) represents plant \(i\)'s share in the portfolios installed capacity.
Portfolio return, \(R\), is given by its mean capacity factor or, equivalently as the installed capacity is fixed, its mean output. This formulation allows to derive the efficient frontier by increasingly restricting the standard deviation (eq. \eqref{eq:MarkSD}), reducing the value of $\sigma_P$. For each value of $\sigma_P$, the model will determine the optimal portfolio (eq. \eqref{eq:MarkObj}). Total installed capacity is kept at a fixed value in all iterations  (eq. \eqref{eq:MarkSumX}). The efficient frontier can consequently be plotted by plotting $R(X_i)$ vs. $\sigma_P$.

By analyzing the objective function and the constraints, it is clear that there are three different portfolio parameters, in which two of them are optimized and can vary at each iteration: standard deviation (eq. \eqref{eq:MarkSD}) and portfolio output (eq. \eqref{eq:MarkObj}).
The remaining parameter --- installed capacity --- is kept at a fixed value in all iterations (eq. \eqref{eq:MarkSumX}).
In that way, it is possible to plot a two-dimensional curve containing the efficient frontier, the non-fixed parameters being the axes.

\hypertarget{improvement-1-cost-minimization}{%
\subsection{Improvement 1: cost minimization}\label{improvement-1-cost-minimization}}
Our first improvement is to minimize costs of the portfolio instead of maximizing capacity factors, as costs can be more directly interpreted by system designers.
First, instead of maximizing generation, as in equation \eqref{eq:MarkObj}, we minimize the portfolio installed capacity, fixing total generation. This particular variation was used in \textcite{novacheckDiversifyingWindPower2017}.
As total generation is directly related to the installed capacity, in principle, it does not matter whether generation is maximized, fixing capacity, or capacity is minimized, fixing generation. Both formulations are equivalent to the maximization of the capacity factor, given by the ratio between mean generation and installed capacity. The first formulation achieves this by maximizing the numerator, and the second one, by minimizing the denominator.
However, as the standard deviation in each formulation relates to a different fixed reference -- either capacity or generation --, in the fixed capacity formulation, portfolios which standard deviation is lower than the standard deviation of the portfolio with highest Sharpe ratio are not part of the efficient frontier in the latter case\footnote{For a demonstration of this, see Appendix \ref{proof}.}.
The Sharpe ratio is the ratio between capacity factor and standard deviation \parencite{shahriariCapacityValueOptimal2018} and represents the slope of the line connecting the origin to the point corresponding to the portfolio's parameters.
Similarly, \textcite{huGeographicalOptimizationVariable2019} and \textcite{thomaidisOptimalManagementWind2016a} used the inverse of this ratio, called coefficient of variation (CV), in their analysis. In the remainder of this text, we will refer to portfolios with lower SD and higher CV than the minimal CV portfolio as ``\textbf{LowSD\_HighCV} portfolios''.

Subsequently, to derive costs, we additionally multiply capacities by their respective unit costs. 
If all candidate technologies or plants have the same costs, this is equivalent to just minimizing capacity, of course. However, under cost differences, a different portfolio will be chosen.  To summarize, the formulation proposed here finds the efficient frontier that achieves the lowest portfolio costs and lowest standard deviations, at the same generation level.
Equations \eqref{eq:CostObj}, \eqref{eq:CostSD}, \eqref{eq:CostIsoGen} and \eqref{eq:CostNonNeg} show the proposed formulation.

\begin{align}
 &\underset{P_{i}}{\text{minimize}}& &C_P(P_{i}) = \sum_{i=1}^{N}P_i C_i \mu_i & & \label{eq:CostObj}
\\
 &\text{subject to} & &\sqrt{\sum_{i=1}^{N}\sum_{j=1}^{N}P_iP_j\sigma_{ij}} \le \sigma_P \label{eq:CostSD}
& \\
 & &&\sum_{i=1}^{N}P_i \mu_i = {K} \label{eq:CostIsoGen} \\
& & & P_{i} \ge 0 \quad \forall i & \label{eq:CostNonNeg}
\end{align}

\(C_P\) is the mean portfolio cost, in \$/MWh, \(C_i\) is plant \(i\) cost in \$/MWh, \(P_i\) is plant \(i\) installed capacity, in MW, and \(\mu_i\) is average output per capacity, i.e., capacity factor of plant \(i\) and \(K\) is an arbitrary fixed value corresponding to the mean portfolio output, in MW.

\hypertarget{improvement-2-demand-correlation}{%
\subsection{Improvement 2: demand correlation}\label{improvement-2-demand-correlation}}

In order to represent demand in the optimization problem, a demand profile is introduced as an additional generator\footnote{From now on, this special generator will be referred as \textbf{DemandGen}.}.
That generator has, however, some special characteristics.
It does not produce electricity, but consumes it. To represent that characteristic, the time series of outputs always has negative values.
Therefore, the correlation between \textbf{DemandGen} and the generation of another plant is the correlation between the demand profile and the generation profile of that plant multiplied by -1. Furthermore, the weight of \textbf{DemandGen} in the portfolio is not a decision variable in the optimization, but it is an input parameter to the model. This weight represents the relative size of \textbf{DemandGen} in comparison to regular generators.

In case of using a small value for \textbf{DemandGen} capacity, its influence will be almost irrelevant.
Conversely, a very high value may overestimate the correlation of generators and the demand.
Therefore, we defined the share of \textbf{DemandGen} to be equal to the peak load value. In our simulations, when using the original formulation, where total capacity is constant, a constraint sets the total portfolio installed capacity equal to the maximum demand value.
In the alternative formulation presented in Section \ref{improvement-1-cost-minimization}, that constraint has to be adapted because generation is kept constant, not capacity. Thus, mean portfolio generation is set equal to mean electricity demand.

To summarize, equations \eqref{eq:CostSDwLoad} and \eqref{eq:CostIsoGenwLoad} represents a cost minimization model which integrates demand: 
\begin{equation}
\sqrt{\sum_{i=1}^{N}\sum_{j=1}^{N}P_iP_j\sigma_{ij} + 2 \sum_{i=1}^{N} P_i P_L \sigma_{iL}} \le \sigma_P \label{eq:CostSDwLoad} 
\end{equation}

\begin{equation}
\sum_{i=0}^{N}P_i \mu_i = P_L \mu_L \label{eq:CostIsoGenwLoad}
\end{equation}

\(P_L\) is the peak load value, \(\sigma_{iL}\) is the covariance between plant \(i\) and \textbf{DemandGen} and \(\mu_L\) is the \textbf{DemandGen} capacity factor, i.e., a negative value representing the relation between mean demand and peak load.
According to equation \eqref{eq:CostIsoGenwLoad}, total mean generation must be equal to mean load.
The equation that represents the objective function \eqref{eq:CostObj} remains unchanged because there is no cost associated to the \textbf{DemandGen} generator.

It should be noted that, in the hypothetical case where a flat load profile is used instead of the real load profile, its variance, correlation and covariance to all the regular generators is 0 ($\sigma_{iL}$ in equation \eqref{eq:CostSDwLoad}).
Thus, the results will be the same as in the formulation that does not use a demand profile. That is expected, as a flat demand profile is an implicit assumption in the traditional model.

\hypertarget{improvement-3-supply-risk}{%
\subsection{Improvement 3: supply risk}\label{improvement-3-supply-risk}}

In the MPT formulation, given two portfolios with the same standard deviation, the one with higher mean output is preferred, regardless of the shape of the output distribution.
Whilst higher generation or lower standard deviation in Gaussian or any other symmetric distributions means higher chance of reaching a minimum output (or balance) level \parencite{scalaPortfolioAnalysisGeographical2019}, that is not necessarily true for the distribution resulting from the portfolio output.

Thus, we propose to incorporate a risk measure that guarantees a minimum energy balance level at a given risk.
Here, we define energy balance at each time step \(t\) as the difference between portfolio output and demand.
Ideally, energy balance would always be a non-negative value because that means that generation is equal or higher than the load at every moment and backup resources would not be necessary.
However, a system based on VRES with this level of security is unrealistic. Therefore, some shortage is accepted and the level of acceptance is given by the risk measure.

Value at Risk (VaR) is one potential risk measure \parencite{grotheSpatialDependenceWind2011}. VaR\textsubscript{\(\beta\)\%} is the \(\beta\)-percentile of the loss distribution. 
However, VaR lacks some properties, such as sub-additivity and convexity \parencite{rockafellarOptimizationConditionalValueatrisk2000}, so it is an imperfect measure to be incorporated into the optimization problem.
Moreover, VaR is based only on the relative number of periods with a negative energy balance. It does not inform about the intensity of those negative balances.
Thus, a more coherent risk measure is the Conditional Value-at-Risk (CVaR) \parencite{sarykalinValueatRiskVsConditional2008}, which is the mean of the lowest \(\beta\) values\footnote{See Appendix \ref{AppendixCvar} for details on the CVaR interpretation and the derivation of the risk constraints \eqref{eq:RiskCVaR}, \eqref{eq:RiskZ} and \eqref{eq:RiskNonNeg}.}. 

The formulation is as follows:

\begin{align}
 \underset{P_{i},Z_{t_{m}},\alpha}{\text{minimize}}&&&C_P(P_{i}) = \sum_{i=1}^{N}P_i C_i \mu_i &
\\
  \text{subject to} &&&\sqrt{\sum_{i=1}^{N}\sum_{j=1}^{N}P_iP_j\sigma_{ij} + 2 \sum_{i=1}^{N} P_i P_L \sigma_{iL}} \le \sigma_P & \label{eq:CVaRSDLoad} 
\\
 &&& \alpha - \frac {\sum_{m = 1}^{M} Z_{t_{m}}}{\beta M} \ge \omega \quad & \label{eq:RiskCVaR}
\\
 &&& \alpha - \sum_{i=1}^{N}(Y_{t_{m},i} P_i) - Y_{t_{m},L} P_L \le Z_{t_{m}} \quad &\forall m ~ in ~ M \subseteq T \label{eq:RiskZ}
  \\
 &&& Z_{t_{m}} \ge 0 \quad &\forall m ~ in ~ M \subseteq T \label{eq:RiskNonNeg}
\end{align}

Equations \eqref{eq:RiskCVaR}, \eqref{eq:RiskZ} and \eqref{eq:RiskNonNeg} are constraints that limit the mean of the lowest \(\beta\) values of energy balance, based on \textcite{rockafellarOptimizationConditionalValueatrisk2000} and \textcite{sarykalinValueatRiskVsConditional2008}.
The model uses samples $Y_{t_{m},i}$ of each plant's output, including \textbf{DemandGen}. The samples are derived from the time series $Y_{t,i}$ which are also used to determine mean generation, variance, and covariances.  
More specifically, \(Y_{t_{m},i}\) is plant \(i\)'s generation over capacity during sample \(m\)'s time step.
Therefore, \(Y_{t_{m},i}P_i\) is plant \(i\)'s generation and \(-Y_{t_{m},L}P_L\) is system demand at time step sample \(t_{m}\).
\(M\) is the total number of samples.
This means, that while the CVaR\textsubscript{$\beta$} is represented exactly by $\text{CVaR}_{\beta}=\alpha - \frac {\sum_{t = 1}^{T} Z_{t}}{\beta T}$, summing only over a subset of $M$ samples from $T$ yields an approximation of $\text{CVaR}_{\beta} \approx \alpha - \frac {\sum_{m = 1}^{M} Z_{t_{m}}}{\beta M}$ instead.
Some attention should be given to the sample size of $M$ because, as each sample adds one constraint in the optimization problem (see eq. \eqref{eq:RiskZ}), values of \(M\) too high may become too computationally expensive, taking a long time to complete.
Values of \(M\) too low may yield an inaccurate CVaR estimation.
In this study, we evaluated different values, and we concluded that \(M\) equal to 3000 samples using Latin Hypercube Sample (LHS)\footnote{Latin Hypercube Sample was used to ensure an evenly distributed sample from years, months and time of day.} techniques are a good compromise between performance and results accuracy and that value was used in all the optimization runs using this model formulation.

In equation \eqref{eq:RiskCVaR}, \(\beta\) and \(\omega\) are input parameters supplied to the model and they represent the confidence level and the balance limit, respectively.
The parameter \(\beta\) refers to the proportion of samples whose average balance must be higher than or equal to \(\omega\).
By setting \(\omega\) to 0, the mean of the \(\beta\) worse balance values will necessarily be a non-negative number.
This means that less than \(\beta\) time steps have energy shortage.
We chose 5\% as the value of \(\beta\) in this work.
The remaining parameters, \(\alpha\) and \(Z_{t_{m}}\), are auxiliary decision variables necessary for the efficient computation of the CVaR by the solver.
At the optimal solution, the \(\alpha\) value obtained corresponds to the VaR\textsubscript{$\beta$}, i.e., the share of \(\beta\) samples that have a higher balance value than the \(\alpha\).

This model computes the installed capacity of each power plant, forming a portfolio that is able to supply system demand fully at the desired level of risk.
However, as long as the share of each power plant is kept equal, the portfolio characteristics are preserved, and it can be applied to any installed capacity level.

\hypertarget{input-data-and-assumptions}{%
\subsection{Input data and assumptions}\label{input-data-and-assumptions}}

In order to assess the models described in the previous sections, we performed a case study based on data of potential Brazilian VRES locations and electricity demand. We assumed a copperplate system, i.e., there are no transmission constraints.
Furthermore, we did not include limits on the total share or installed capacity of each power plant or technology, even though this kind of constraint could be easily implemented. As the goal in this study is to evaluate the impact of the proposed model improvements rather than to obtain a realistic and detailed portfolio for the Brazilian grid, we chose to not use those constraints. That way, our results are easier to compare and understand.

\hypertarget{time-series-for-generation-and-load}{%
\subsubsection{Time series for generation and load}\label{time-series-for-generation-and-load}}

The generation data used in the case study was obtained based on the simulation of wind power generation \parencite{gruberAssessingGlobalWind2019} and PV power generation \parencite{ramirezcamargoSimulationMultiannualTime2020} from the latest generation of reanalysis climate data sets. From an initial data set consisting of 168 locations for PV and 415 for wind plants, we first eliminated locations that had similar time series.
For that purpose, we identified pairs of locations that had correlation greater than 99\%. From each pair, we kept only one of them. This process was repeated until correlation between remaining plants was always below or equal to the threshold. This allowed to reduce the dimension of the optimization problem and avoid redundant plants. After that process, 104 PV plants and 55 wind plants remained as candidates in the case study.

Demand time series historical data is available on the website of the Brazilian independent system operator (ONS) \parencite{onsCurvaCargaHoraria2019}. In order to eliminate the growth component from the time series, we normalized each time step by the maximum value in neighboring days.
All data, demand and outputs for each location, have hourly time steps and cover the period from January 1999 to August 2017, totaling to 163,611 hours.

Figure \ref{fig:PlantsCorrel} shows how Pearson correlation values are distributed among each type of power plant. Correlations among photovoltaic power plants are always high, ranging from 0.60 to 0.99.
On the other hand, considering only wind power plants, pairwise Pearson correlations vary greatly, ranging from -0.05 to 0.99.
Most wind power plants are negatively correlated to PV and, when correlation is positive, its absolute value is low, which shows that the combination of both technologies is advantageous to achieve portfolios with a low standard deviation.
Finally, PV is, in general, positively correlated to load, although at a low level.
This can be explained by the fact that the observed peak load in Brazil, most of the time, occurs during the afternoon due to the increased demand for air conditioning systems during this time of the day.

\begin{figure}
\centering
\includegraphics{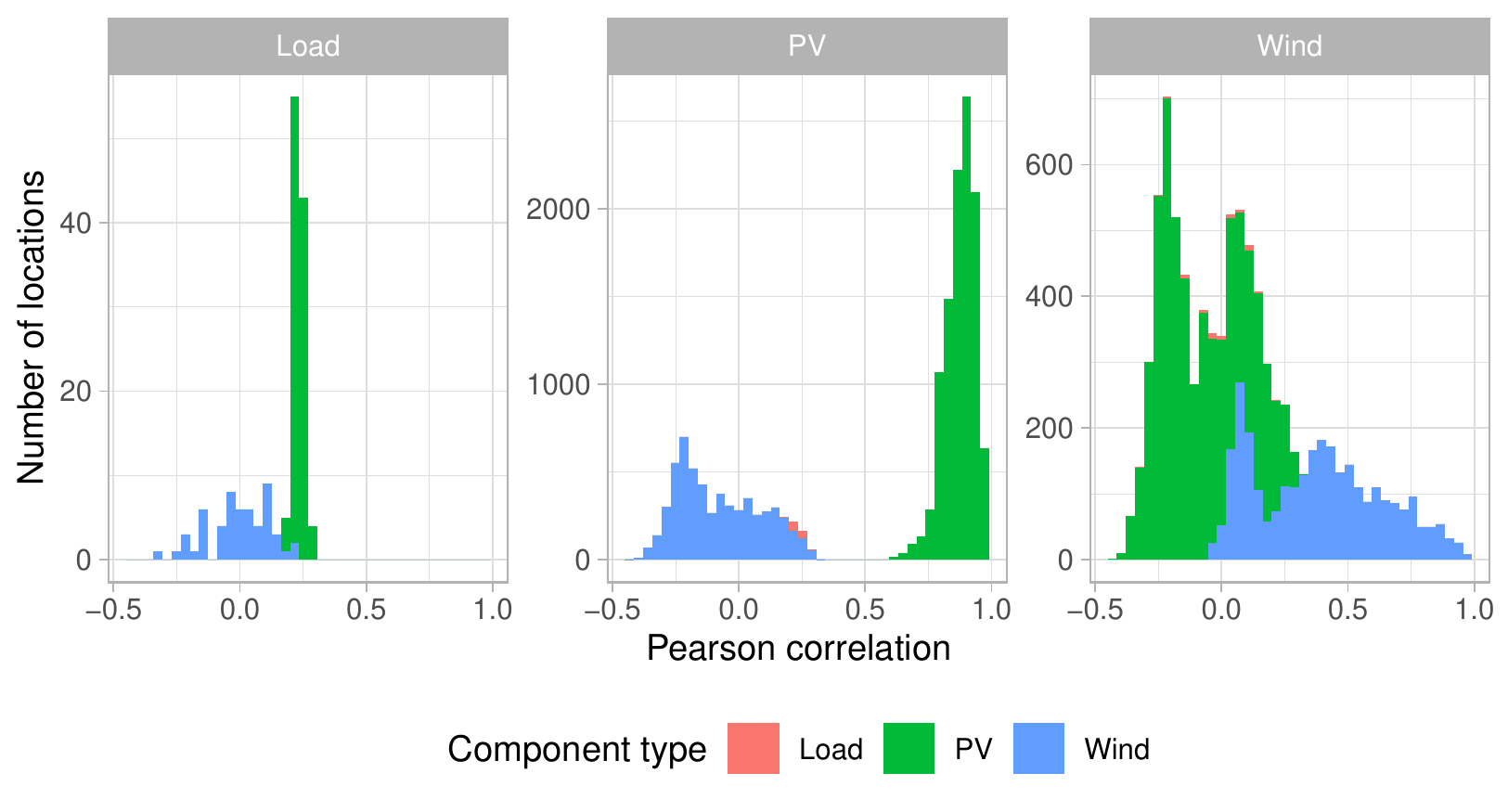}
\caption{\label{fig:PlantsCorrel}Histogram plot which shows the distribution of pairwise Pearson correlations.  Each facet shows the plant types separately: load (left), photovoltaics (center) and wind power plants (right).}
\end{figure}

Figure \ref{fig:distXcor} shows how correlation among power plants based on the same technology is related to the distance between them. There is an inverse relation between distance and correlation, when fixing one technology and this relation is more pronounced for wind plants. That agrees with the expectations as, even though PV generation is affected by different weather conditions, the radiation daily cycle only changes slightly based on the latitude and, mainly, longitude.

\begin{figure}[p!]
\centering
\includegraphics{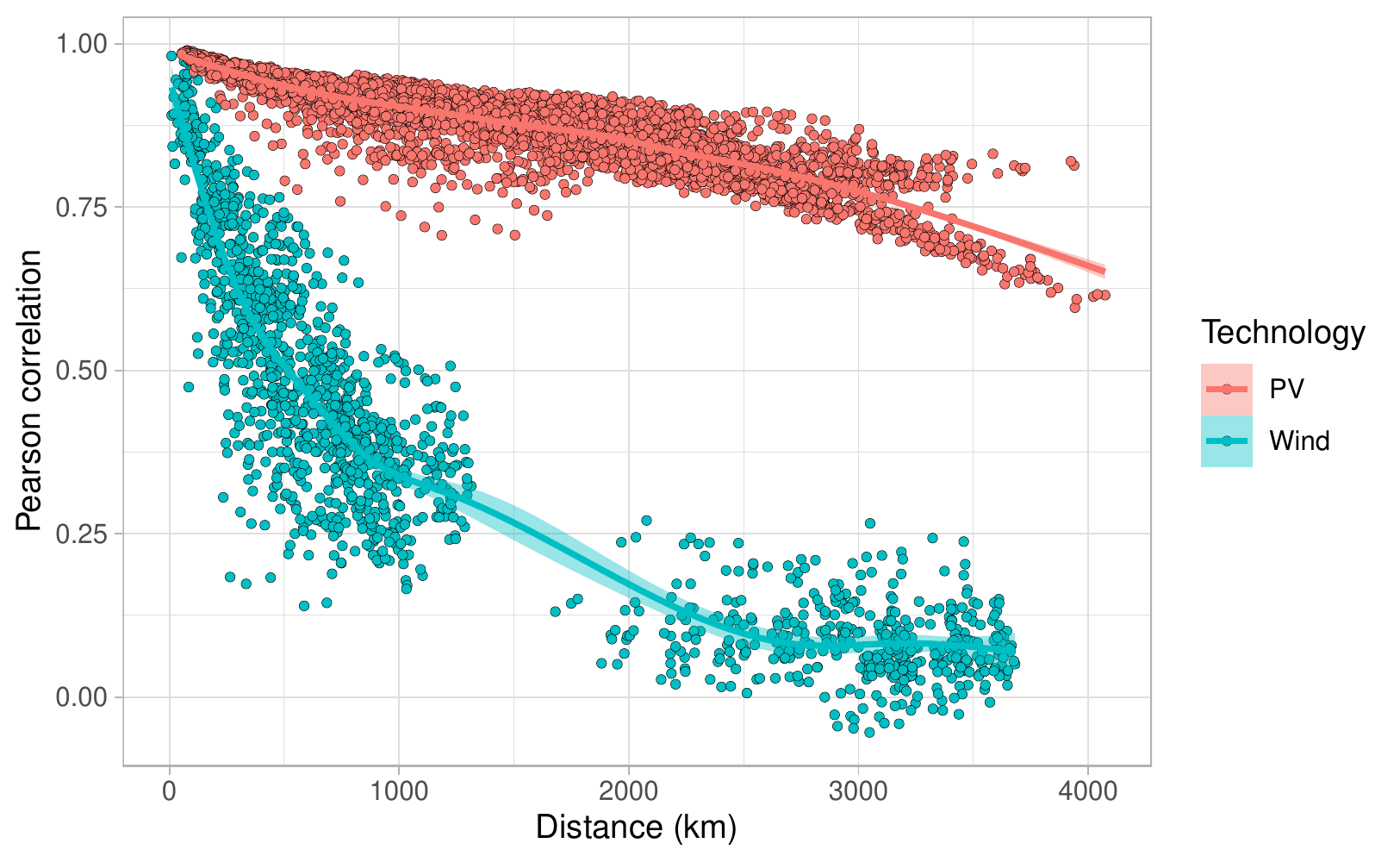}
\caption{\label{fig:distXcor}Scatterplot of the pairwise Pearson correlation of all locations of the same technology versus the geographic distance in kilometers}
\end{figure}

\hypertarget{costs}{%
\subsubsection{Costs}\label{costs}}

The power plant costs were estimated based on input data for the Brazilian Decennial Energy Plan -- PDE 2029 \parencite{epeCustoMarginalExpansao2019}. Table \ref{tab:PlantCosts} shows those values in Brazilian Reais (BRL)\footnote{In the year PDE 2029 was released, the average exchange rate was 3.88 BRL/USD.}.
Individual power plant costs can vary due to particular characteristics such as location, terrain steepness and existing infrastructure.
In this case, it would be possible to attribute individual investment and O\&M costs to each plant.
However, as a simplifying assumption, we used the same investment and O\&M costs for plants of the same technology, according to Table \ref{tab:PlantCosts}.

Based on investment and O\&M costs, we calculated the levelized costs of each plant, in \$/MWh. Therefore plants of the same technology with higher capacity factor will have lower costs per unit of energy. As the difference in the levelized cost of those technologies is relatively small, we performed a sensitivity analysis in which PV costs are significantly lower: half of the costs shown in Table \ref{tab:PlantCosts}.

\begin{longtable}[]{@{}lrr@{}}
\caption{\label{tab:PlantCosts} Power plant costs. Source: \textcite{epeCustoMarginalExpansao2019}.}\tabularnewline
\toprule
Type & Investment cost (BRL/kW) & O\&M costs (BRL/kW.y) \\
\midrule
\endfirsthead
\toprule
Type & Investment cost (BRL/kW) & O\&M costs (BRL/kW.y) \\
\midrule
\endhead
Wind & 4800 & 90 \\
PV & 3500 & 50 \\
\bottomrule
\end{longtable}

\hypertarget{models-summary}{%
\subsection{Scenarios}\label{models-summary}}

We run different scenarios -- in each scenario, a different formulation of the optimization model will be used. Each scenario has its own short name, as shown in Table \ref{tab:ModelDescription}.
The first character of the code refers to the combination of objective function and constraint used and the second character flags whether the employed demand curve is assumed constant or based on real, observed demand data.

\begin{longtable}[]{@{}
  >{\raggedright\arraybackslash}p{(\columnwidth - 8\tabcolsep) * \real{0.19}}
  >{\raggedright\arraybackslash}p{(\columnwidth - 8\tabcolsep) * \real{0.25}}
  >{\raggedright\arraybackslash}p{(\columnwidth - 8\tabcolsep) * \real{0.15}}
  >{\raggedright\arraybackslash}p{(\columnwidth - 8\tabcolsep) * \real{0.20}}
  >{\raggedright\arraybackslash}p{(\columnwidth - 8\tabcolsep) * \real{0.19}}@{}}
\caption{\label{tab:ModelDescription} Models description.}\tabularnewline
\toprule
Model & Description & Objective function & Fixed constraint & Load curve \\
\midrule
\endfirsthead
\toprule
Model & Description & Objective function & Fixed constraint & Load curve \\
\midrule
\endhead
Trad\_Flat & Traditional model. & Maximize generation & Portfolio capacity & Flat demand or nonexistent \\
Trad\_Obs & Trad\_Flat with observed demand profile. & Maximize generation & Portfolio capacity & Observed demand \\
Cost\_Flat & Technologies costs. & Minimize cost & Portfolio generation & Flat demand or nonexistent \\
Cost\_Obs & Cost\_Flat with observed demand profile. & Minimize cost & Portfolio generation & Observed demand \\
Cost\_Flat\_lcpv & Cost\_Flat with low-cost photovoltaics & Minimize cost & Portfolio generation & Flat demand or nonexistent \\
CVaR\_Flat & Energy balance at fixed risk level. & Minimize cost & Average energy balance of lowest 5\% time steps & Flat demand or nonexistent \\
CVaR\_Obs & CVaR\_Flat with observed demand profile. & Minimize cost & Average energy balance of lowest 5\% time steps & Observed demand \\
\bottomrule
\end{longtable}

In order to find the upper and lower bound for the efficient frontier, a method similar to the one used by \textcite{huGeographicalOptimizationVariable2019} was implemented. We first executed the model removing the constraint that limits the standard deviation, to obtain the portfolio with the highest variance.
Then, the variance was minimized, in order to find the lowest variance portfolio. We subsequently executed the optimization 51 times, in each run changing the maximum value of the standard deviation (\(\sigma_P\)) by equally distanced steps, from a lower to a higher value. To run these models, we used the CPLEX solver.
In order to prepare the data and to analyze optimization results, the software R \parencite{rcoreteamLanguageEnvironmentStatistical2020} were used. The code and data files will be made available upon publication at github.com/tuberculo/VRES-Portfolio.

\hypertarget{diversification}{%
\subsection{Diversification}\label{diversification}}

One important aspect is to verify the level of diversity of the portfolios in the efficient frontier in the different model formulations. For that purpose, we will use three different indices: the mean geographic distance, weighted by plant generation; the mean Euclidean distance, also weighted by generation share; and the Herfindahl-Hirschman Index (HHI), as in \textcite{shahriariCapacityValueOptimal2018}.

\hypertarget{mean-geographical-distance}{%
\subsubsection{Mean geographical distance}\label{mean-geographical-distance}}

Based on the location of the plants and their relative share, it is possible to calculate the mean distance for each portfolio, as shown in Equation \eqref{eq:GeoDiv}.
\(D_{ij}\) is the shortest distance between plants \(i\) and \(j\), in kilometers, and \(w_i\) is the generation share\footnote{Plant generation over total portfolio generation.} of power plant \(i\).
Therefore, higher values for this index show that the plants are, on average, more apart from each other.

\begin{equation}
GD = \sum_{i = 1}^N \sum_{j = 1}^N w_i w_j D_{ij} 
\label{eq:GeoDiv} 
\end{equation}

\hypertarget{mean-non-geographical-distance}{%
\subsubsection{Mean Euclidean distance}\label{mean-non-geographical-distance}}

It is also possible to use the Euclidean distance (Equation \eqref{eq:EucDiv}) of the generation profile, instead of the geographical distance, in order to calculate a diversity index.
In this case the distance (\(E_{ij}\)) is the sum of the square of the differences of two power plants generation at each time step, as described in equation \eqref{eq:DefDist}.
The higher the value for this index, the higher the diversification.

\begin{equation}
E_{ij} = \sqrt{\sum_{t = 1}^{T} (i_t - j_t)^2}
\label{eq:DefDist} 
\end{equation}

\begin{equation}
ED = \sum_{i = 1}^N \sum_{j = 1}^N w_i w_j E_{ij} 
\label{eq:EucDiv} 
\end{equation}

Compared to HHI (presented below), this index has the advantage to measure one additional property of diversity: Disparity (degree of differences)\footnote{The other properties are Balance (evenness in contributions) and Variety (number of elements in mix) \parencite{awerbuchAnalyticalMethodsEnergy2008}.}, i.e., it does not only measure how many elements there are in the set and their relative weights, it does also measure how different each element is from the others.

\hypertarget{hhi}{%
\subsubsection{HHI}\label{hhi}}

The Herfindahl-Hirschman Index (HHI) is the sum of the squared share of each component (equation \eqref{eq:HHI}). Therefore, a HHI equal to 1 means that the portfolio contains only one plant. Lower HHI values indicate high diversification.

\begin{equation}
HHI = \sum_{i = 1}^N w_i^2
\label{eq:HHI}
\end{equation}

It is easy to show that the minimum HHI value occurs when all plant's weights are equally distributed and that minimum value is \(1/N\), where \(N\) is the total number of elements.

This index yields lower values for highly diversified sets, differently from the previous indexes. We will therefore use the inverse of the HHI to ensure comparability. This transformation has the additional advantage of showing the effective number of power plants in the portfolio \parencite{jostEntropyDiversity2006}.
For instance, a portfolio which HHI is equal to 0.1 has a diversity equivalent to a portfolio of 10 equally weighted plants.

\hypertarget{results-and-discussion}{%
\section{Results and discussion}\label{results-and-discussion}}

\hypertarget{methodologies-improvements}{%
\subsection{Impacts of model adaptations}\label{methodologies-improvements}}

In this section, we analyze the impact of each proposed improvement on model results, highlighting the main differences in results obtained from each formulation. As the models have different fixed constraints, i.e., they have different sizing in terms of installed capacities, mean generation and other parameters, comparing results directly could be misleading. Therefore, in order to compare them more easily, standard deviation and costs are normalized.
We consequently divided the value of each portfolio's power output standard deviation by the portfolio installed capacity in order to obtain the normalized standard deviation. This is equivalent to a comparison of all portfolios with the same installed capacity.
In the same way, portfolio equivalent costs represents the weighted average of the plants LCOE or, equivalently, the total costs divided by the total energy produced in \$/MWh.

\hypertarget{improvement-1-costs}{%
\subsubsection{Improvement 1: cost minimization}\label{improvement-1-costs}}

The resulting efficient frontiers from the reference model (\textbf{Trad\_Flat}) and the model that minimizes cost (\textbf{Cost\_Flat}) are shown in figures \ref{fig:CompareFrontierCost} and \ref{fig:CompareFrontierCost2}.

As a reference of how portfolios in the frontier compare to individual plants, dots representing capacity factor, standard deviation and levelized cost of each individual candidate plant were added.
The curve for the \textbf{Cost\_Flat} model input parameters sensitivity (\textbf{Cost\_Flat\_lcpv}), which has lower costs for PV plants are shown only in Figure \ref{fig:CompareFrontierCost} because, as plant costs are not equal between optimization runs, their portfolios costs are not comparable.
The point where the coefficient of variation (CV) in \textbf{Trad\_Flat} model is minimal, is highlighted in both figures.
As discussed in Section \ref{improvement-1-cost-minimization} and in Appendix \ref{proof}, there is no portfolio with lower standard deviation than that in the \textbf{Cost\_Flat} model.

\begin{figure}
\includegraphics[width=0.5\linewidth]{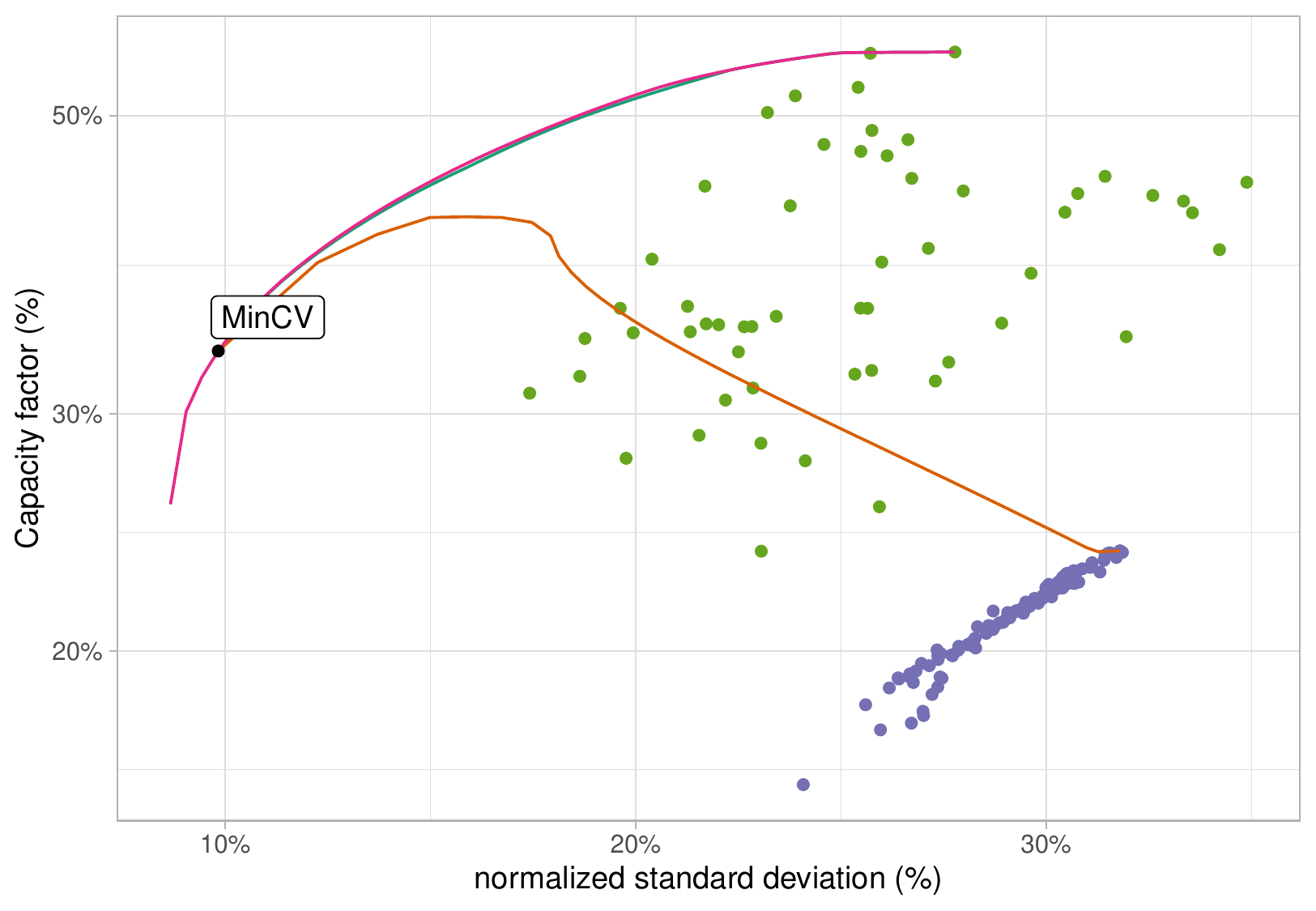} \includegraphics[width=0.5\linewidth]{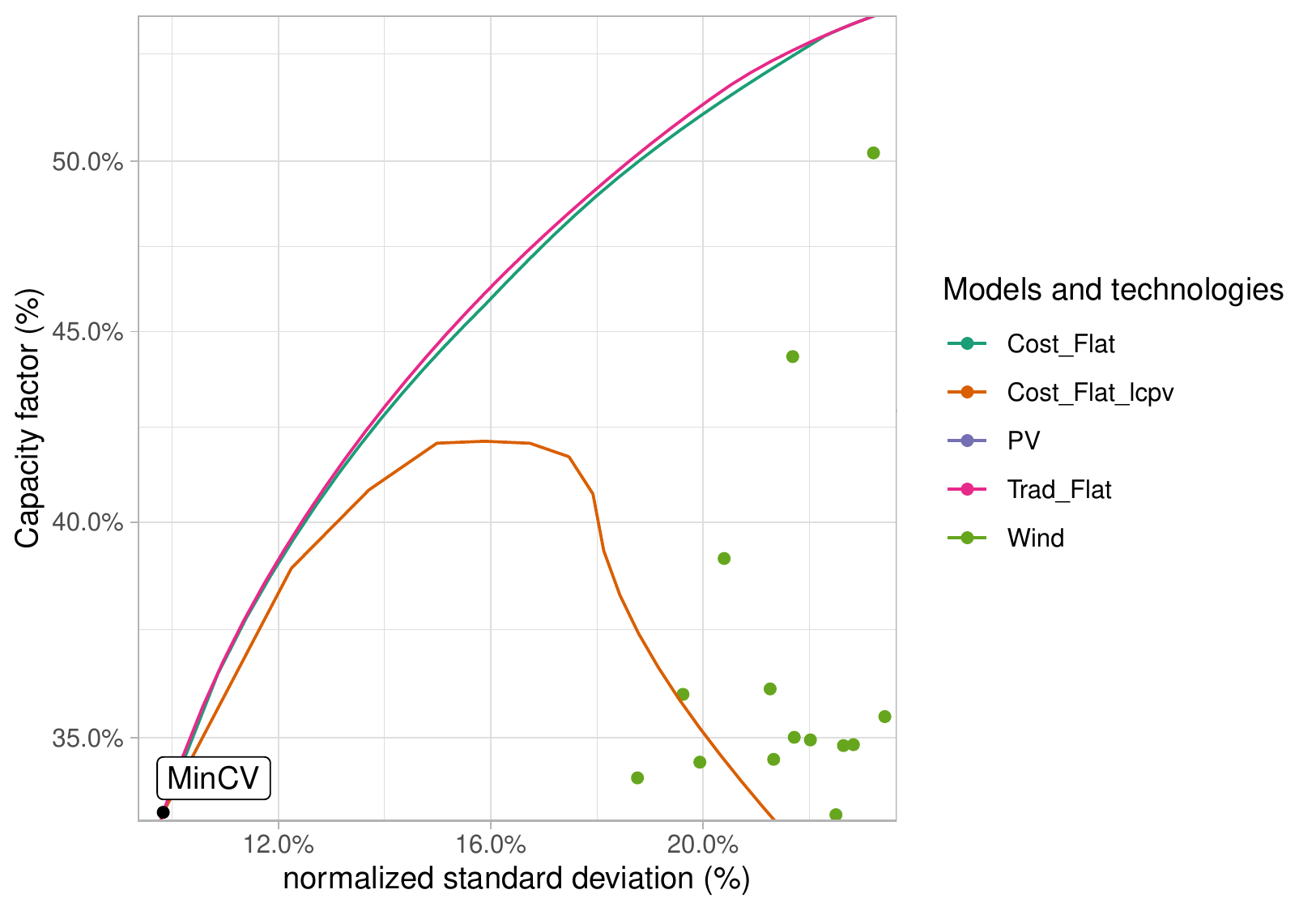} \caption{Efficient frontier by capacity factor. On the right, a zoomed version highlighting the region where the curves from \textbf{Trad\_Flat} and \textbf{Cost\_Flat} model differ.}\label{fig:CompareFrontierCost}
\end{figure}

\begin{figure}
\includegraphics[width=0.5\linewidth]{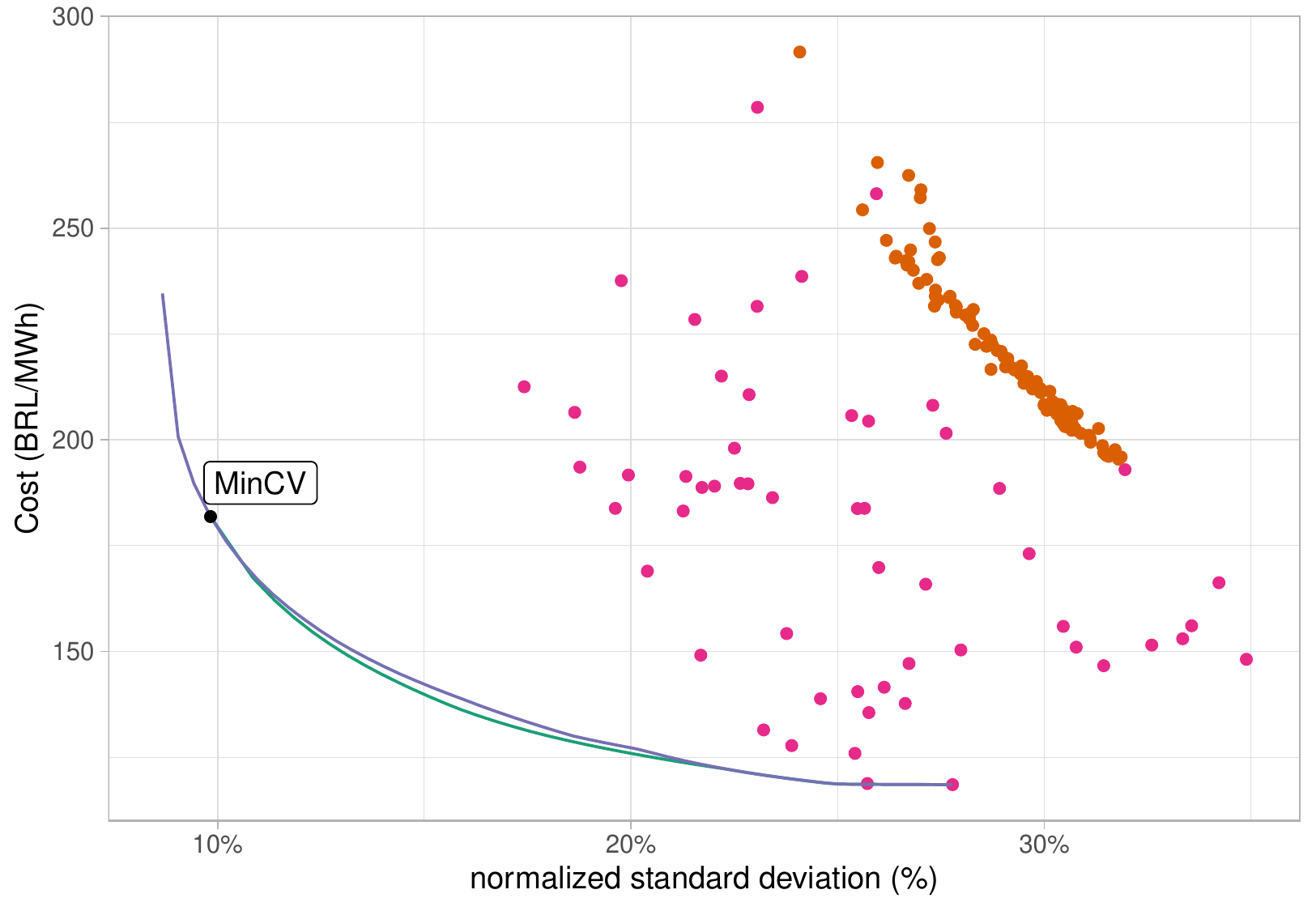} \includegraphics[width=0.5\linewidth]{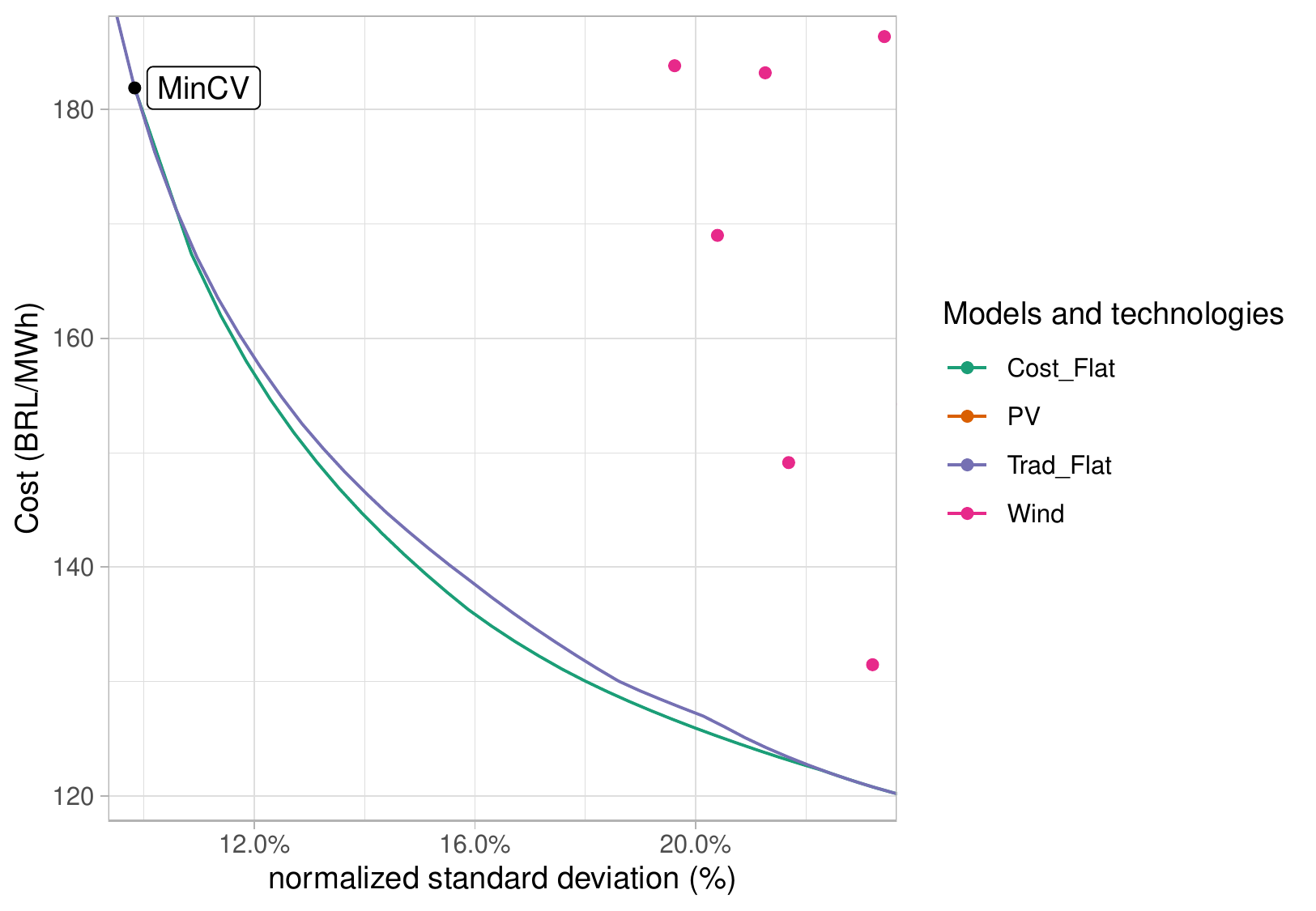} \caption{Efficient frontier by portfolio levelized cost. On the right, a zoomed version highlighting the region where the curves from \textbf{Trad\_Flat} and \textbf{Cost\_Flat} model differ.}\label{fig:CompareFrontierCost2}
\end{figure}

As expected, the results show that portfolio costs from \textbf{Cost\_Flat} model are lower or equal to the ones obtained in \textbf{Trad\_Flat} model and \textbf{Trad\_Flat} portfolio capacity factors are higher or equal to the ones in \textbf{Cost\_Flat} and \textbf{Cost\_Flat\_lcpv} models.
Nevertheless, as costs for both technologies are similar, the results do not differ much between the \textbf{Trad\_Flat} and the \textbf{Cost\_Flat} model.
However, in \textbf{Cost\_Flat\_lcpv}, the results are different for all levels of standard deviation, except in the point where SD is minimized, because standard deviation does not depend on plant costs.

To analyze how the portfolio composition changes along the curve in each model, Figure \ref{fig:CompareRunsSidebySide} shows the relative share of the power plants that compose a portfolio.
As there are too many plants in the data set, we aggregated them by similar capacity factors to simplify visualization.
In comparison to the reference model (\textbf{Trad\_Flat}), the PV share increases slightly when minimizing costs using standard costs (\textbf{Cost\_Flat}) because they have lower investment costs than wind plants.
When PV costs are reduced (\textbf{Cost\_Flat\_lcpv}), portfolio composition changes more intensely and
PV plants tend to substitute wind plants as SD increases.

\begin{figure}
\centering
\includegraphics{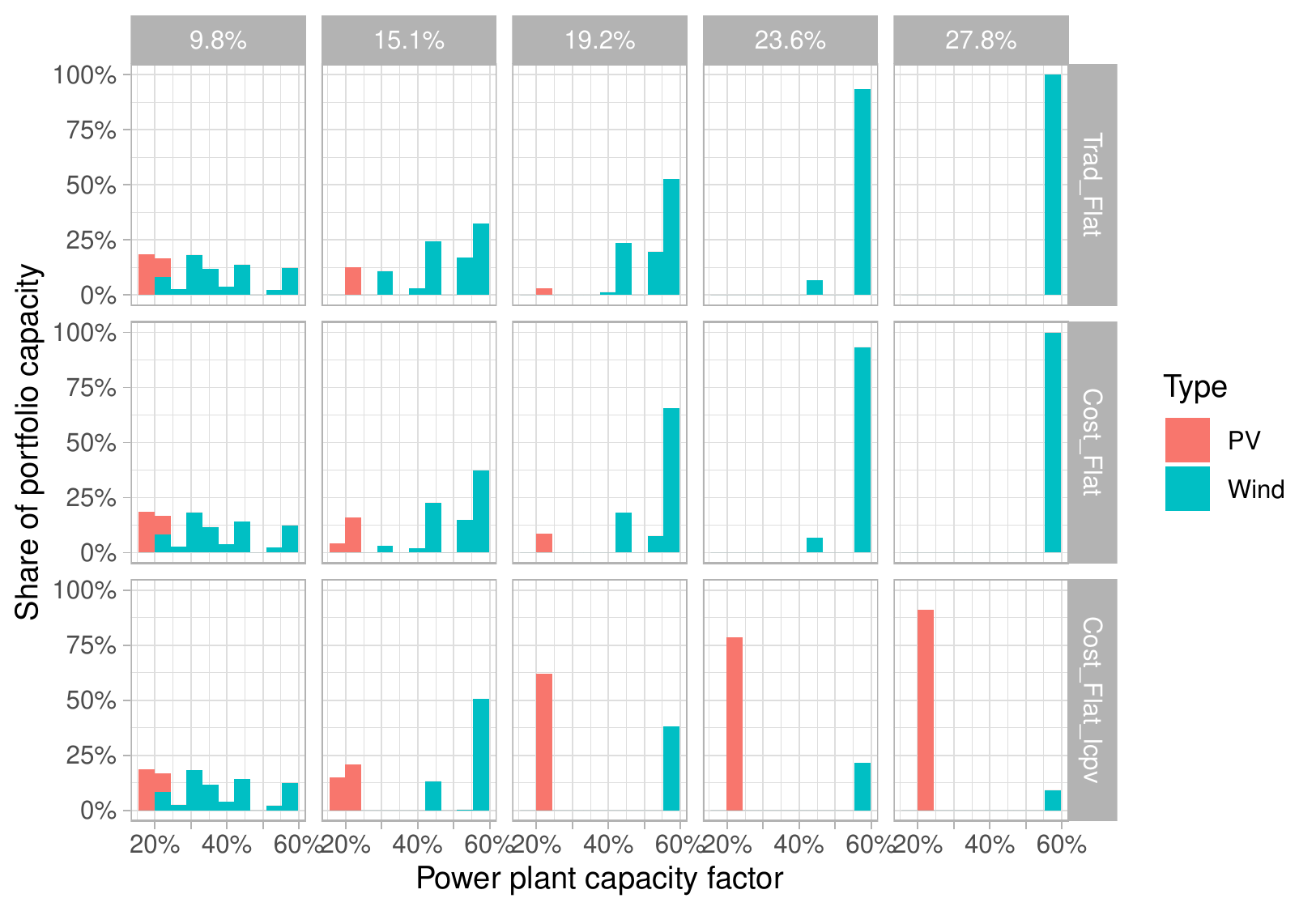}
\caption{\label{fig:CompareRunsSidebySide}Comparison of portfolio composition. Each facet is a portfolio and the height of the bars is the share of the capacity of plants in that range of capacity factor. Therefore, the sum of the bars in a facet is always 1. Different models are in different rows and the values in the label of each column is the rounded value of portfolio standard deviation.}
\end{figure}

In summary, changing the objective function from maximizing capacity factor to minimizing costs improves the results by incorporating information relevant to the decision maker. It also has the advantage of discarding portfolios with too low capacity factor (LowSD\_HighCV portfolios). Those portfolios, counterintuitively, have a higher resulting standard deviation when compared to other portfolios with the same expected generation and higher capacity factor. As this formulation change doesn't add any complexity to the computational problem, there is no disadvantage of using it instead of model \textbf{Trad\_Flat} .

\hypertarget{improvement-2-demand-profile}{%
\subsubsection{Improvement 2: demand profile}\label{improvement-2-demand-profile}}

As shown in Section \ref{time-series-for-generation-and-load}, Brazilian aggregated demand has, for most of the plants, higher correlation to PV output profile than to wind power output profile.
Consequently, incorporating demand into the model, as described in Section \ref{improvement-2-demand-correlation}, has the effect of increasing the relative share of PV plants in portfolios obtained in models \textbf{Trad}, \textbf{Cost} and \textbf{CVaR} using observed demand when compared to their equivalent models using flat demand (Figure \ref{fig:PVShareDemand}).
However, this increase does not occur in portfolios with high standard deviation, as they already have a tendency to have low PV shares.

\begin{figure}
\centering
\includegraphics{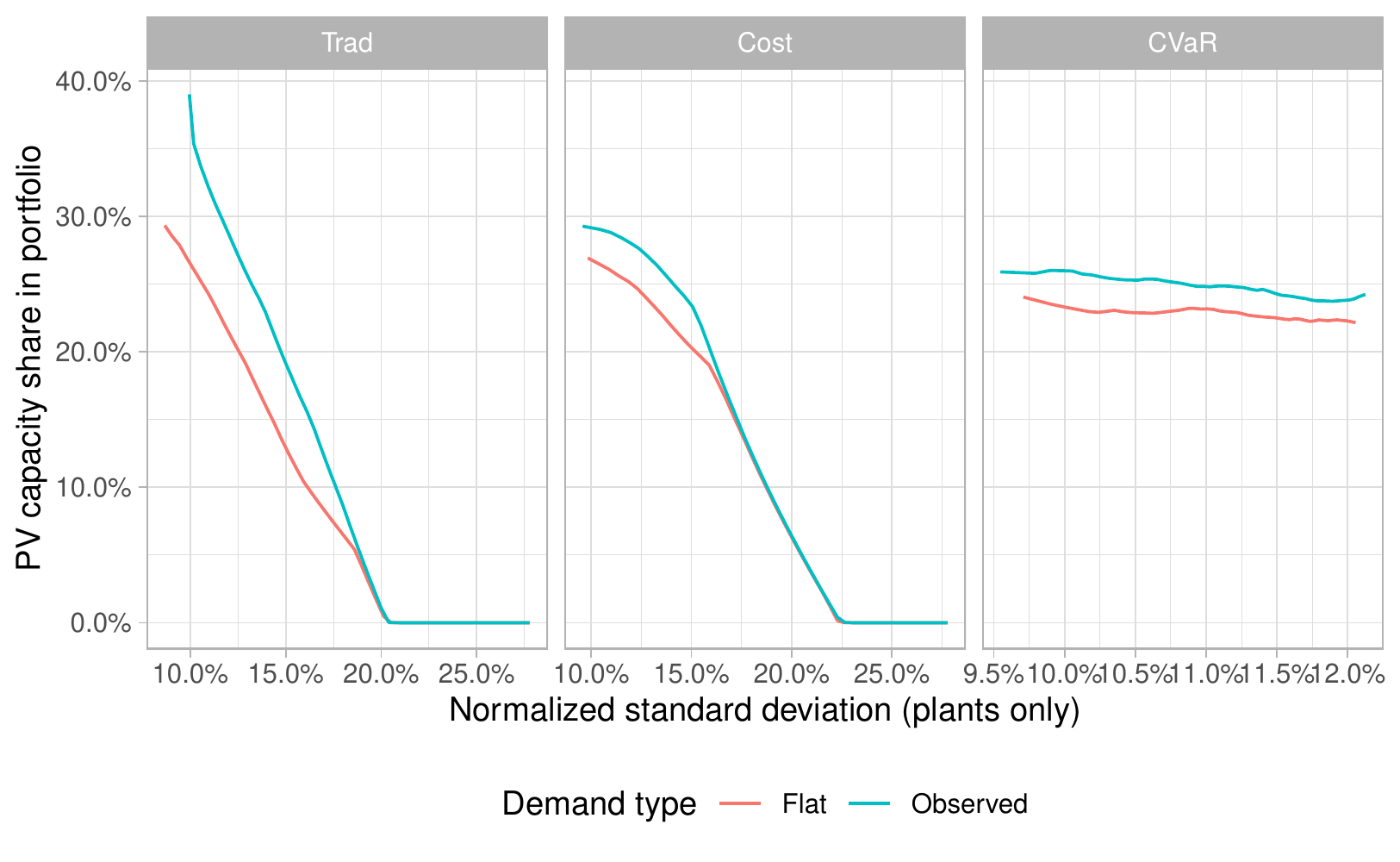}
\caption{\label{fig:PVShareDemand}PV share in portfolios for each different model formulation.}
\end{figure}

\hypertarget{improvement-3-supply-risk}{%
\subsubsection{Improvement 3: supply risk}\label{improvement-3-supply-risk-results}}

In Figure \ref{fig:FrontierRisk}, the results from \textbf{CVaR\_Flat} model --- the one which adds a lower threshold for the energy balance of the worst performing time steps --- are compared to the results from the previous models.
The first important finding is that the curve obtained in the \textbf{CVaR\_Flat} model does not overlap in any point with the curves from the \textbf{Trad\_Flat} or the \textbf{Cost\_Flat} model.
Portfolios resulting from the \textbf{CVaR\_Flat} model have a lower mean capacity factor (and higher cost) for the same standard deviation.
Nevertheless, they are still preferred because they have higher values at the lower percentiles of the output distribution when compared to the portfolios of the other two models. This implies that no portfolio in the efficient frontier obtained from the traditional MPT formulation is optimal in terms of ensuring a firm output level. Therefore, if the goal is to maximize firm output, then applying MPT and selecting the best portfolios will probably not result in obtaining the best portfolios in terms of firm output.

Another outstanding characteristic in Figure \ref{fig:FrontierRisk} is the length of the curve from the \textbf{CVaR\_Flat} model.
It is, in comparison to the other models, much shorter.
There are no portfolios with high standard deviation in the \textbf{CVaR\_Flat} model. We conclude that portfolios with high SD values, by having a more wide output distribution, would require an increase in the total installed capacity, elevating the mean costs above the costs from lower SD portfolios.

\begin{figure}
\centering
\includegraphics{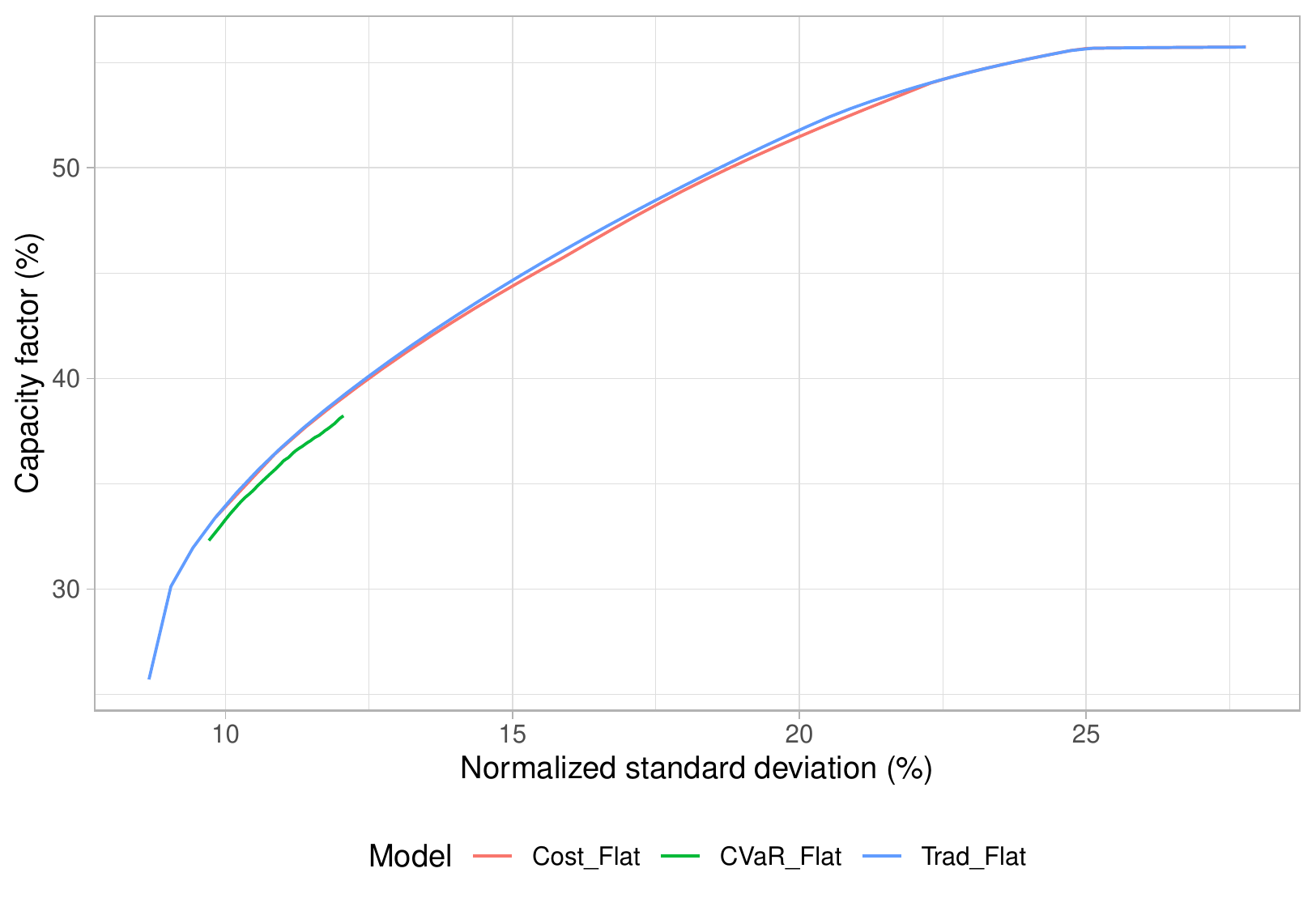}
\caption{\label{fig:FrontierRisk}Efficient frontier by capacity factor, including \textbf{Trad\_Flat}, \textbf{Cost\_Flat} and \textbf{CVaR\_Flat} models.}
\end{figure}

The downside of this formulation is the higher computational time required to execute it when compared to the \textbf{Trad\_Flat} and \textbf{Cost\_Flat} models.
In our runs, we observed an increase in execution time of around 50 times. That ratio can vary based on system configuration and parameters used, such as the number of potential plants and number of samples used to determine CVaR.

\hypertarget{diversification-1}{%
\subsection{Diversification}\label{diversification-1}}

In the \textbf{Trad} and \textbf{Cost} models the highest variance portfolio is composed of only one plant as there is no constraint limiting the share of each technology or location: it is the one with the highest capacity factor.
When different plants are included in a portfolio, SD reduces at the cost of lower capacity factors and higher LCOE. Thus portfolios with higher standard deviation tend to have a smaller diversity and vice versa. On the other hand, when CVaR constraints are used (models \textbf{CVaR\_Flat} and \textbf{CVaR\_Obs}), all resulting portfolios are highly diversified (Figure \ref{fig:Diversification}).
Hence, in order to have low generation costs and a firm output level, diversification plays an important role.

For the \textbf{CVaR} model, the inverse of HHI decreases as SD increases, indicating a decrease in diversity. However, the other metrics, ED and GD, remain at a high level for every SD.
Therefore, even though higher SD concentrates output share in some plants, those plants have a very distinct output profile, which explains high ED level, and are far from each other, which explains high GD level.

The mean geographical distance among power plants, for all portfolios in the \textbf{CVaR\_Flat} and \textbf{CVaR\_Obs} models, is higher than 1600 km. This is an indication of how integrating different regions located far from each other is an advantage when the goal is to allow for a steadier generation profile at lower costs.
The combination of power plants located in places with different climate conditions, increases the chance that the combined generation is at least equal to demand.
Likewise, the spreading of PV power plants from East to West has the effect of maintaining the output from this source available for more hours during a day.
Naturally, as this case study was performed based on a copperplate assumption and without constraints on the size of each plant, it is possible that those portfolios would not be feasible in more realistic conditions. Nevertheless, those results show how important geographical spreading is to smooth output profile.

\begin{figure}
\centering
\includegraphics{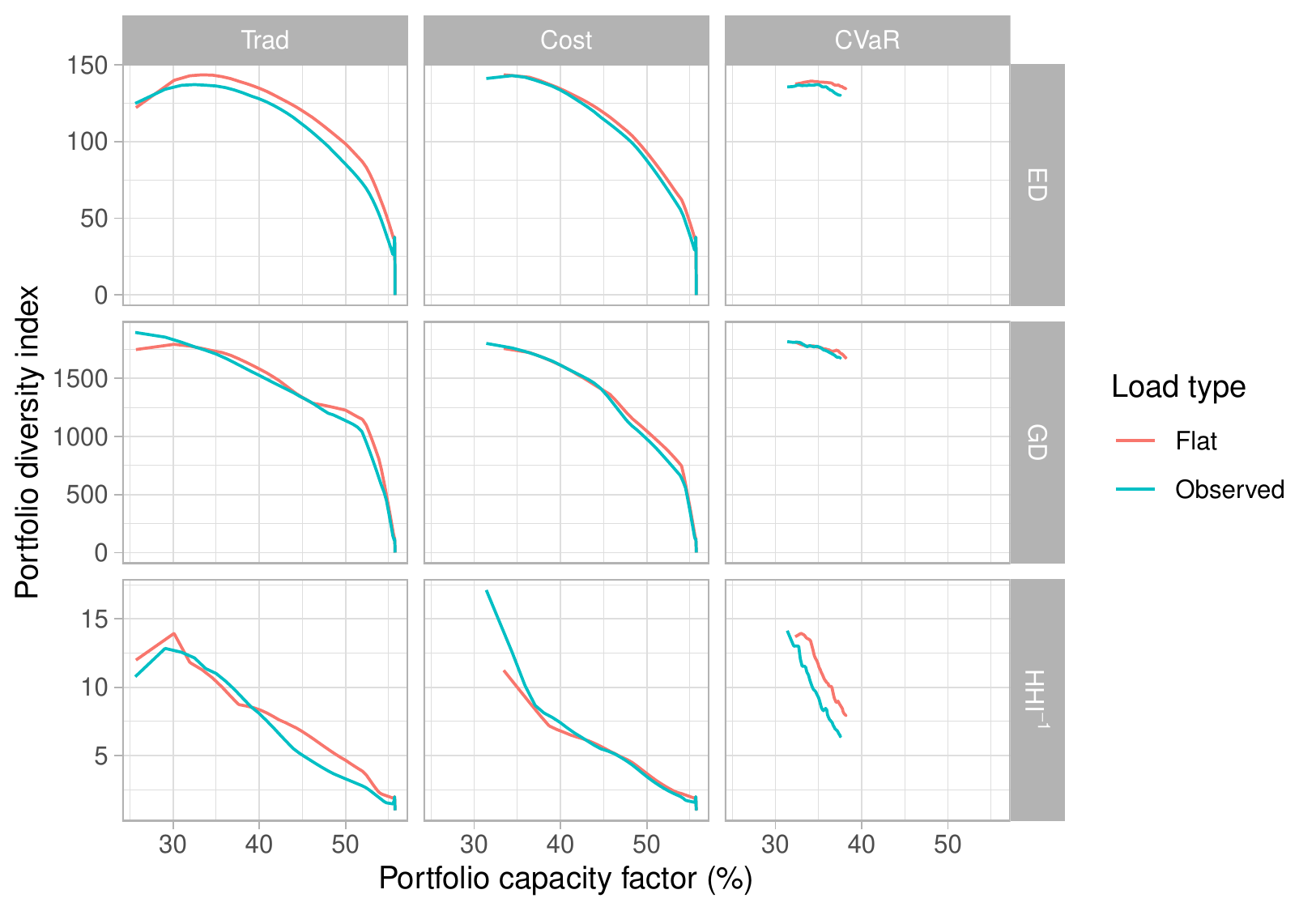}
\caption{\label{fig:Diversification}Portfolio diversification for each methodology. Each facet column is a model type. Each facet row is a diversification index. ED (Euclidean distance) is dimensionless, GD (Geographical Distance) is in kilometers and HHI (Herfindahl-Hirschman Index) is also dimensionless. For ED, GD and inverse of HHI, higher values mean more diversification.}
\end{figure}

In this context, wind power expansion in Brazil is relatively low diversified, concentrating around 90\% of installed capacity in Northeast region, in few states. The share of wind power plant in Northeast region in all portfolios of \textbf{CVaR\_Flat} and \textbf{CVaR\_Obs} scenarios is much lower, around 55\%. Therefore, exploring more the wind potential of other regions in Brazil could improve the firm energy level of the system.

\hypertarget{validation-same-risk-level-comparison}{%
\subsection{Same risk level comparison}\label{validation-same-risk-level-comparison}}

Until now, all comparisons among different models were performed by normalizing all portfolios to have the same total installed capacity.
In this section, a different approach is used. Here, we compare portfolios based on necessary installed capacity to attend the risk criteria used in models \textbf{CVaR\_Flat} and \textbf{CVaR\_Obs}.
That is, keeping the portfolio proportion of each plant unchanged, we find the lowest installed capacity that satisfies the following constraint: mean energy balance at the 5\% lowest values is equal to 0.
In this analysis, risk levels were calculated for the whole time series, not only for the 3000 samples used in the portfolio optimization problem.
For this reason, for portfolios obtained from the \textbf{CVaR\_Flat} and \textbf{CVaR\_Obs} models, the installed capacity values are not equal to the ones found in the portfolio optimization, they are, however, quite similar (less than 2\% difference in all cases).

With those calculated installed capacity values, every portfolio is equally capable of attending demand at the required risk level. Nevertheless, some of them would be more expensive than others.
Figure \ref{fig:IsoRiskMultiplier} shows the mean cost per megawatt-hour of demand for the different portfolios.
As expected, \textbf{CVaR} models are indeed the models that provide the portfolios with lowest costs.
The other models have similar costs in a narrow region of the curve, beginning at the minimum CV point the in \textbf{Trad} model and at the minimum SD point in the \textbf{Cost} model.
However, outside this region, system cost increases quickly and reaches values as high as approximately 3.5 times the lowest cost.

In the curve of the \textbf{Trad} model, some  portfolios are clearly dominated, they under-perform both in terms of standard deviation and cost. These are the \textbf{LowSD\_HighCV} portfolios, as discussed in Section \ref{improvement-1-cost-minimization}.
The expected capacity factor from those portfolios is low and to supply demand at the risk criteria used, more installed capacity is needed, causing the SD and costs to increase more than for other portfolios.

When the demand profile is considered, the results are similar except that the \textbf{Trad\_Obs} model performs worse than others in the region with low standard deviation.
That can be explained by the fact that the relative weight of \textbf{DemandGen} is smaller in this case, as discussed in Section \ref{improvement-2-demand-correlation}.

\begin{figure}
\centering
\includegraphics{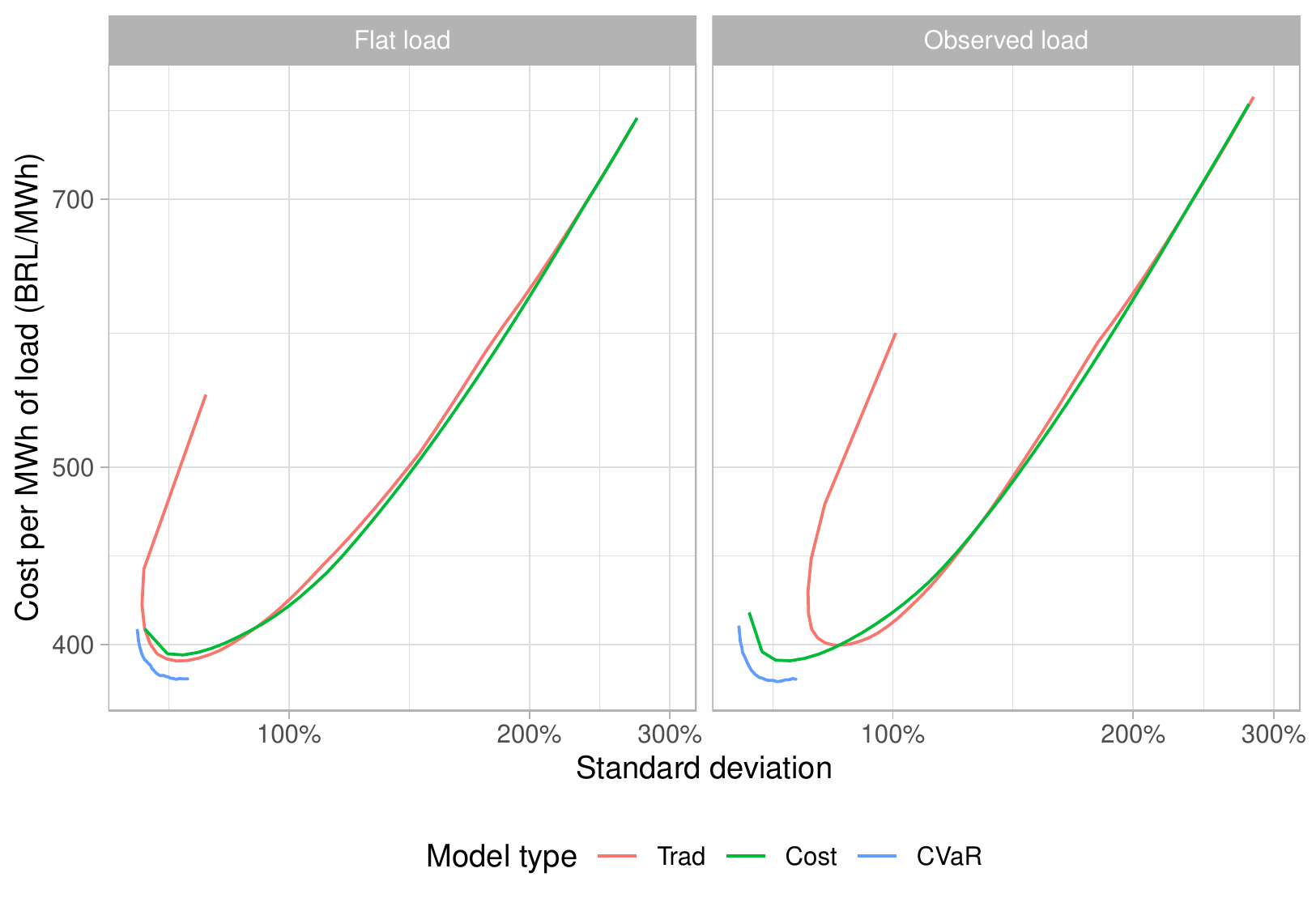}
\caption{\label{fig:IsoRiskMultiplier}Each curve represents portfolio cost normalized by the necessary installed capacity to supply demand at the risk criteria used in the \textbf{CVaR\_Obs} model (balance mean of 5\% worse time steps is 0). Cost values in y-axis are per MWh of system demanded. The standard deviation values, in the x-axis, are normalized to the new installed capacity. For a better visualization, the plot area is limited at 800 BRL/MWh instead of the maximum value of 1374.40 BRL/MWh.}
\end{figure}

We selected three portfolios for each model to evaluate in a more detailed way how energy balance behaves in each portfolio normalized by CVaR risk criteria.
The portfolios selected were the ones with lowest, intermediate and highest standard deviation for each model that uses demand (\textbf{Trad\_Obs}, \textbf{Cost\_Obs} and \textbf{CVaR\_Obs}).
In all cases, the expected excess energy is more than 100\% of demand, i.e., mean output from portfolios is more than two times higher than mean demand (Figure \ref{fig:EnergyBalanceDistribution}).
This value could be reduced by using storage systems or by accepting a higher shortage probability (increasing \(\beta\) value) and providing more backup power plants.
Those possibilities are not assessed in this model and can be investigated in future research.

Portfolios with higher standard deviation require higher mean output to be able to achieve the required risk level.
However, the \textbf{Trad\_Obs} and \textbf{Cost\_Obs} models have a much higher increase in expected output with higher standard deviation.
This is because, as SD increases, portfolio output gets more dispersed, reducing the firm capacity factor from the portfolio. Consequently, more installed capacity is needed and mean output increases. 

\begin{figure}
\centering
\includegraphics{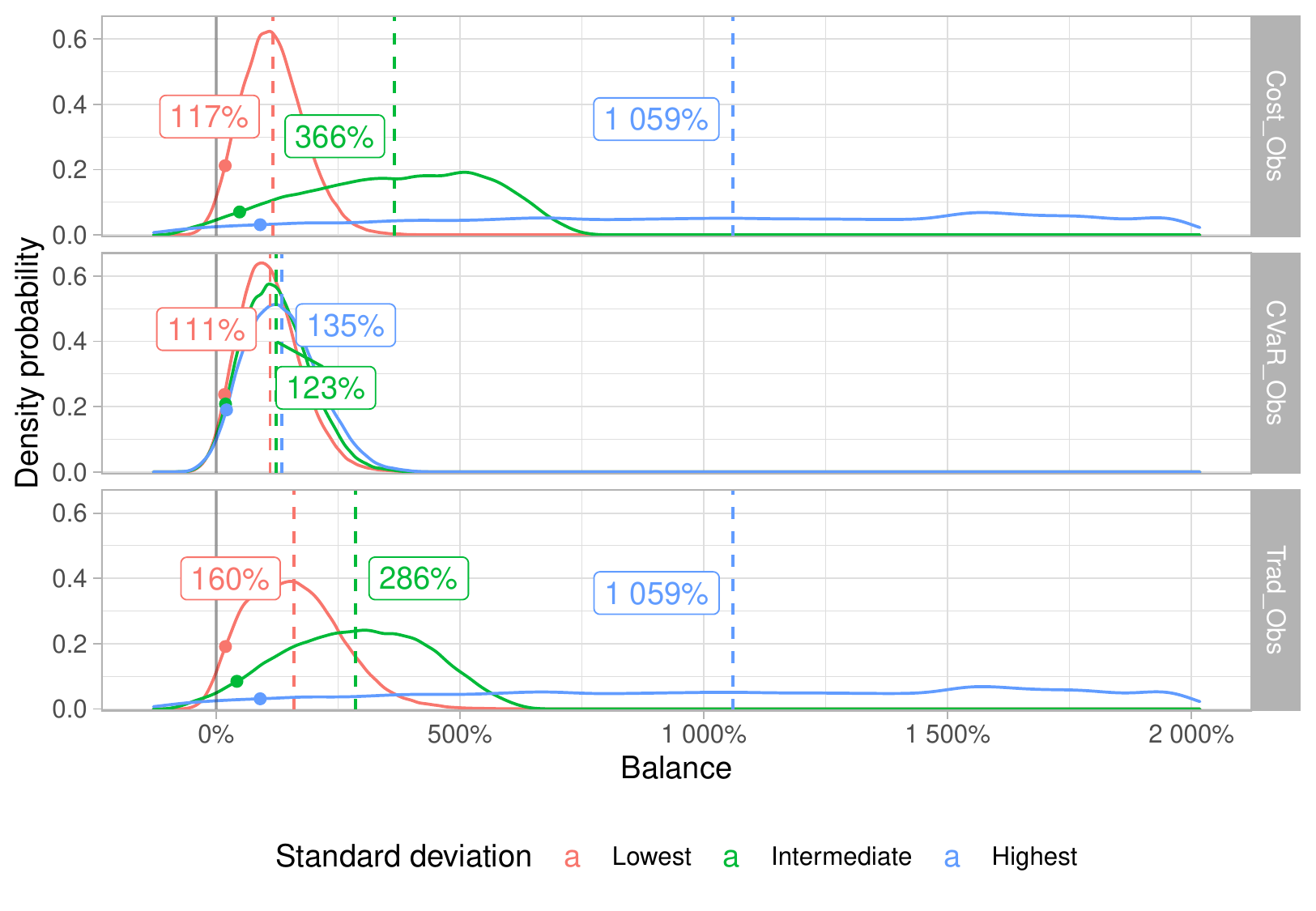}
\caption{\label{fig:EnergyBalanceDistribution}Distribution of energy balance over mean demand for three portfolios in models \textbf{Trad\_Obs}, \textbf{Cost\_Obs} and \textbf{CVaR\_Obs}. Balance values are in percent and are proportional to mean demand. Vertical dashed lines and labels show the expected value of the energy balance for each portfolio. The solid colored points are the VaR\textsubscript{5\%} of each portfolio.}
\end{figure}

\hypertarget{limitations-and-further-research}{%
\subsection{Limitations and further research}\label{limitations-and-further-research}}

In this work, there are some limitations that could be improved in further research.
 First, this study is based on VRES only. Thus, the interactions between VRES and other types of generation or storage are not explored. These additional technologies include hydro power, biomass, concentrating solar power and storage systems. The incorporation of those technologies and even fossil thermal power plants would help to evaluate the interaction between technologies and probably, portfolios with lower real-world costs could be found. However, in order to incorporate non-VRES technologies, some adaptations and simplifications are necessary because those technologies are, to some extent, controllable. Therefore, to calculate their variance, correlation and other necessary parameters, some strong assumptions would have to be made.

  The expected output from VRES power plants is based on historical climate profiles for each plant and depends on its location and technology. However, output patterns could change in the future, for instance, due to climate change.
  This kind of risk is not considered in the data we use.
  
  In this study, we used a single configuration for CVaR constraints in the \textbf{CVaR\_Flat} and \textbf{CVaR\_Obs} models, i.e., only one set of values was used for \(\beta\) and \(\omega\), and we maintained the minimization of SD as one of the objectives. This way, the comparison to the other models is straight-forward. However, different parametrizations could be used, depending on the desired goal. It is possible to jointly use several CVaR constraints, with different parameters \(\beta\) and \(\omega\) that would give more control over the desired shape of the energy balance distribution. Another possibility is to replace one of the optimization goals: instead of SD minimization, use CVaR minimization.

  The time series used in this is study have hourly resolution.
  Using more detailed data could better capture short-term variation in plants output and load.
  \textcite{shahriariCapacityValueOptimal2018} found that when different time resolutions are used, the optimal portfolio composition in MPT models change for the Eastern United States grid. Thus, a similar analysis considering the formulation proposed here would be valuable.

\hypertarget{conclusions}{%
\section{Conclusions}\label{conclusions}}

In this paper, we proposed and assessed several improvements to the application of MPT in renewable electricity systems. Those improvements incorporate three important aspects that were lacking in previous studies, notably: (i) focusing on minimizing  portfolio costs instead of maximizing capacity, (ii) considering correlation between load and generation profiles and (iii) limiting shortage risks via the inclusion of a CVaR measure.

We found out that incorporating plant costs in the formulation better captures important information, as investment decisions are not only based on the capacity factors of plants and the interaction among them (covariance), but also on their costs. Depending on the relative costs of candidate plants, the resulting efficient frontier can be completely different when this change is added to the formulation. The exception is the portfolio with lowest standard deviation. Additionally, a positive side effect of this proposed formulation is that portfolios which SD is lower than the SD of the portfolio with minimum CV value are removed from the resulting efficient frontier.
Those portfolios are not desirable because their reduced standard deviation is mostly explained by their reduced capacity factor.

We also proposed to include system demand in portfolios, represented as a power plant with negative generation. As expected, this change increases the share of plants that have relatively higher correlation with system demand.
It is worth noting that those changes in formulation do not have any downside in terms of computational complexity.

Another, more important, improvement proposed and assessed was changing the constraint that sets a fixed value for all portfolios: from fixed installed capacity or mean output to fixed CVaR at a given level of risk.
We found out that none of the portfolios in the efficient frontier from traditional MPT are in the optimal set after this change. Furthermore, all portfolios obtained have a low standard deviation when compared to the standard deviation range of the efficient frontier in the traditional MPT.
This shows that high standard deviation portfolios are not suited to provide firm energy output.
Therefore, the benefits obtained from this formulation justify its use even considering the increase in computational burden to solve the optimization model.

The portfolios obtained after all changes in formulation were performed, have in common a high diversification level, being composed by a combination of many different and spatially separated plants. This fact highlights how well diversified portfolios are important to ``smooth'' aggregated output and to guarantee a firm output level.

Many of the plants that are part of the optimal set are not among plants with lowest levelized costs.
If plants were built under LCOE assumptions only, they would not be built because there are cheaper options. However, those plants play an important role in the portfolio as they help to reduce risk, complementing other plants.
Therefore, developing ways to incentive the construction of those plants is important to better explore the resources of a region in an optimal way.

\hypertarget{acknowledgments}{%
\section*{Acknowledgments}\label{acknowledgments}}
\addcontentsline{toc}{section}{Acknowledgments}

This study was financed in part by the Coordenação de Aperfeiçoamento de Pessoal de Nível Superior - Brasil (CAPES) - Finance Code 001.
C.K. gratefully acknowledges financial support from the Anniversary Fund of the Oesterreichische Nationalbank (OeNB), 18306. J.S., C.K. and P.R. acknowledge  the financial support of the European  Research  Council  (‘reFUEL’  ERC-2017-STG  758149). A.P. acknowledges the support of CNPq and FAPERJ.

\section*{Competing interests}
No competing interests to declare.

\printbibliography

@book{awerbuchAnalyticalMethodsEnergy2008,
  title = {Analytical Methods for Energy Diversity and Security: Portfolio Optimization in the Energy Sector: A Tribute to the Work of {{Dr Shimon Awerbuch}}},
  shorttitle = {Analytical Methods for Energy Diversity and Security},
  editor = {Awerbuch, Shimon and Bazilian, Morgan and Roques, Fabien A.},
  date = {2008},
  publisher = {{Elsevier}},
  location = {{Amsterdam ; London}},
  isbn = {978-0-08-056887-4},
  number = {12},
  pagetotal = {321},
  series = {Elsevier Global Energy Policy and Economics Series}
}

@article{bogdanovRadicalTransformationPathway2019,
  title = {Radical Transformation Pathway towards Sustainable Electricity via Evolutionary Steps},
  author = {Bogdanov, Dmitrii and Farfan, Javier and Sadovskaia, Kristina and Aghahosseini, Arman and Child, Michael and Gulagi, Ashish and Oyewo, Ayobami Solomon and de Souza Noel Simas Barbosa, Larissa and Breyer, Christian},
  date = {2019-03-06},
  journaltitle = {Nature Communications},
  volume = {10},
  pages = {1077},
  publisher = {{Nature Publishing Group}},
  issn = {2041-1723},
  doi = {10.1038/s41467-019-08855-1},
  url = {https://www.nature.com/articles/s41467-019-08855-1},
  urldate = {2020-08-17},
  issue = {1},
  langid = {english},
  number = {1},
  options = {useprefix=true}
}

@article{chuppOptimalWindPortfolios2012,
  title = {Optimal {{Wind Portfolios}} in {{Illinois}}},
  author = {Chupp, Benjamin A. and Hickey, Emily and Loomis, David G.},
  date = {2012-01},
  journaltitle = {The Electricity Journal},
  volume = {25},
  pages = {46--56},
  issn = {10406190},
  doi = {10.1016/j.tej.2012.01.002},
  url = {http://linkinghub.elsevier.com/retrieve/pii/S1040619012000036},
  urldate = {2018-03-29},
  langid = {english},
  number = {1}
}

@article{cunhaDesigningElectricityGeneration2015,
  title = {Designing Electricity Generation Portfolios Using the Mean-Variance Approach},
  author = {Cunha, Jorge and Ferreira, Paula Varandas},
  date = {2015},
  journaltitle = {International Journal of Sustainable Energy Planning and Management},
  volume = {4},
  pages = {17--30},
  doi = {10.5278/ijsepm.2014.4.3}
}

@article{degeilhQuantitativeApproachWind2011,
  title = {A Quantitative Approach to Wind Farm Diversification and Reliability},
  author = {Degeilh, Yannick and Singh, Chanan},
  date = {2011-02-01},
  journaltitle = {International Journal of Electrical Power \& Energy Systems},
  shortjournal = {International Journal of Electrical Power \& Energy Systems},
  volume = {33},
  pages = {303--314},
  issn = {0142-0615},
  doi = {10.1016/j.ijepes.2010.08.027},
  url = {http://www.sciencedirect.com/science/article/pii/S0142061510001638},
  urldate = {2020-10-06},
  langid = {english},
  number = {2}
}

@article{dellano-pazEnergyPlanningModern2017,
  title = {Energy Planning and Modern Portfolio Theory: {{A}} Review},
  shorttitle = {Energy Planning and Modern Portfolio Theory},
  author = {deLlano- Paz, Fernando and Calvo-Silvosa, Anxo and Antelo, Susana Iglesias and Soares, Isabel},
  date = {2017-09},
  journaltitle = {Renewable and Sustainable Energy Reviews},
  volume = {77},
  pages = {636--651},
  issn = {13640321},
  doi = {10.1016/j.rser.2017.04.045},
  url = {http://linkinghub.elsevier.com/retrieve/pii/S136403211730552X},
  urldate = {2018-02-08},
  langid = {english},
  options = {useprefix=true}
}

@article{delucchiProvidingAllGlobal2011,
  title = {Providing All Global Energy with Wind, Water, and Solar Power, {{Part II}}: {{Reliability}}, System and Transmission Costs, and Policies},
  shorttitle = {Providing All Global Energy with Wind, Water, and Solar Power, {{Part II}}},
  author = {Delucchi, Mark A. and Jacobson, Mark Z.},
  date = {2011-03-01},
  journaltitle = {Energy Policy},
  shortjournal = {Energy Policy},
  volume = {39},
  pages = {1170--1190},
  issn = {0301-4215},
  doi = {10.1016/j.enpol.2010.11.045},
  url = {http://www.sciencedirect.com/science/article/pii/S0301421510008694},
  urldate = {2020-09-23},
  langid = {english},
  number = {3}
}

@article{drakeWhatExpectGreater2007,
  title = {What to Expect from a Greater Geographic Dispersion of Wind Farms?—{{A}} Risk Portfolio Approach},
  shorttitle = {What to Expect from a Greater Geographic Dispersion of Wind Farms?},
  author = {Drake, Ben and Hubacek, Klaus},
  date = {2007-08-01},
  journaltitle = {Energy Policy},
  shortjournal = {Energy Policy},
  volume = {35},
  pages = {3999--4008},
  issn = {0301-4215},
  doi = {10.1016/j.enpol.2007.01.026},
  url = {http://www.sciencedirect.com/science/article/pii/S030142150700033X},
  urldate = {2020-07-25},
  langid = {english},
  number = {8}
}

@report{epeCustoMarginalExpansao2019,
  title = {Custo {{Marginal}} de {{Expansão}} Do {{Setor Elétrico Brasileiro}} – {{Metodologia}} e {{Cálculo}} - 2019},
  author = {{EPE}},
  date = {2019-09-17},
  institution = {{EPE}},
  location = {{Rio de Janeiro}},
  url = {https://www.epe.gov.br/sites-pt/publicacoes-dados-abertos/publicacoes/PublicacoesArquivos/publicacao-423/topico-482/NT_CME_EPE_DEE-NT-057_2019-r0.pdf},
  urldate = {2020-10-26},
  number = {EPE-DEE-NT-057/2019-r0}
}

@article{grotheSpatialDependenceWind2011,
  title = {Spatial Dependence in Wind and Optimal Wind Power Allocation: {{A}} Copula-Based Analysis},
  shorttitle = {Spatial Dependence in Wind and Optimal Wind Power Allocation},
  author = {Grothe, Oliver and Schnieders, Julius},
  date = {2011-09-01},
  journaltitle = {Energy Policy},
  shortjournal = {Energy Policy},
  volume = {39},
  pages = {4742--4754},
  issn = {0301-4215},
  doi = {10.1016/j.enpol.2011.06.052},
  url = {http://www.sciencedirect.com/science/article/pii/S0301421511005088},
  urldate = {2020-10-02},
  langid = {english},
  number = {9}
}

@article{gruberAssessingGlobalWind2019,
  title = {Assessing the {{Global Wind Atlas}} and Local Measurements for Bias Correction of Wind Power Generation Simulated from {{MERRA}}-2 in {{Brazil}}},
  author = {Gruber, Katharina and Klöckl, Claude and Regner, Peter and Baumgartner, Johann and Schmidt, Johannes},
  date = {2019-12-15},
  journaltitle = {Energy},
  shortjournal = {Energy},
  volume = {189},
  pages = {116212},
  issn = {0360-5442},
  doi = {10.1016/j.energy.2019.116212},
  url = {http://www.sciencedirect.com/science/article/pii/S0360544219319073},
  urldate = {2020-01-07},
  langid = {english}
}

@article{haegelTerawattscalePhotovoltaicsTransform2019,
  title = {Terawatt-Scale Photovoltaics: {{Transform}} Global Energy},
  shorttitle = {Terawatt-Scale Photovoltaics},
  author = {Haegel, Nancy M. and Atwater, Harry and Barnes, Teresa and Breyer, Christian and Burrell, Anthony and Chiang, Yet-Ming and Wolf, Stefaan De and Dimmler, Bernhard and Feldman, David and Glunz, Stefan and Goldschmidt, Jan Christoph and Hochschild, David and Inzunza, Ruben and Kaizuka, Izumi and Kroposki, Ben and Kurtz, Sarah and Leu, Sylvere and Margolis, Robert and Matsubara, Koji and Metz, Axel and Metzger, Wyatt K. and Morjaria, Mahesh and Niki, Shigeru and Nowak, Stefan and Peters, Ian Marius and Philipps, Simon and Reindl, Thomas and Richter, Andre and Rose, Doug and Sakurai, Keiichiro and Schlatmann, Rutger and Shikano, Masahiro and Sinke, Wim and Sinton, Ron and Stanbery, B. J. and Topic, Marko and Tumas, William and Ueda, Yuzuru and van de Lagemaat, Jao and Verlinden, Pierre and Vetter, Matthias and Warren, Emily and Werner, Mary and Yamaguchi, Masafumi and Bett, Andreas W.},
  date = {2019-05-31},
  journaltitle = {Science},
  volume = {364},
  pages = {836--838},
  publisher = {{American Association for the Advancement of Science}},
  issn = {0036-8075, 1095-9203},
  doi = {10.1126/science.aaw1845},
  url = {https://science.sciencemag.org/content/364/6443/836},
  urldate = {2020-07-29},
  eprint = {31147512},
  eprinttype = {pmid},
  langid = {english},
  number = {6443}
}

@article{haleIntegratingSolarFlorida2018,
  title = {Integrating Solar into {{Florida}}'s Power System: {{Potential}} Roles for Flexibility},
  shorttitle = {Integrating Solar into {{Florida}}'s Power System},
  author = {Hale, Elaine T. and Stoll, Brady L. and Novacheck, Joshua E.},
  date = {2018-08-01},
  journaltitle = {Solar Energy},
  shortjournal = {Solar Energy},
  volume = {170},
  pages = {741--751},
  issn = {0038-092X},
  doi = {10.1016/j.solener.2018.05.045},
  url = {http://www.sciencedirect.com/science/article/pii/S0038092X1830478X},
  urldate = {2020-12-15},
  langid = {english}
}

@article{huGeographicalOptimizationVariable2019,
  title = {Geographical Optimization of Variable Renewable Energy Capacity in {{China}} Using Modern Portfolio Theory},
  author = {Hu, Jing and Harmsen, Robert and Crijns-Graus, Wina and Worrell, Ernst},
  date = {2019-11-01},
  journaltitle = {Applied Energy},
  shortjournal = {Applied Energy},
  volume = {253},
  pages = {113614},
  issn = {0306-2619},
  doi = {10.1016/j.apenergy.2019.113614},
  url = {http://www.sciencedirect.com/science/article/pii/S0306261919312887},
  urldate = {2020-09-21},
  langid = {english}
}

@article{jacobsonMatchingDemandSupply2018a,
  title = {Matching Demand with Supply at Low Cost in 139 Countries among 20 World Regions with 100\% Intermittent Wind, Water, and Sunlight ({{WWS}}) for All Purposes},
  author = {Jacobson, Mark Z. and Delucchi, Mark A. and Cameron, Mary A. and Mathiesen, Brian V.},
  date = {2018-08-01},
  journaltitle = {Renewable Energy},
  shortjournal = {Renewable Energy},
  volume = {123},
  pages = {236--248},
  issn = {0960-1481},
  doi = {10.1016/j.renene.2018.02.009},
  url = {http://www.sciencedirect.com/science/article/pii/S0960148118301526},
  urldate = {2020-09-23},
  langid = {english}
}

@article{jostEntropyDiversity2006,
  title = {Entropy and Diversity},
  author = {Jost, Lou},
  date = {2006},
  journaltitle = {Oikos},
  volume = {113},
  pages = {363--375},
  issn = {1600-0706},
  doi = {10.1111/j.2006.0030-1299.14714.x},
  url = {https://onlinelibrary.wiley.com/doi/abs/10.1111/j.2006.0030-1299.14714.x},
  urldate = {2021-01-16},
  langid = {english},
  number = {2}
}

@article{juraszReviewComplementarityRenewable2020,
  title = {A Review on the Complementarity of Renewable Energy Sources: {{Concept}}, Metrics, Application and Future Research Directions},
  shorttitle = {A Review on the Complementarity of Renewable Energy Sources},
  author = {Jurasz, J. and Canales, F. A. and Kies, A. and Guezgouz, M. and Beluco, A.},
  date = {2020-01-01},
  journaltitle = {Solar Energy},
  shortjournal = {Solar Energy},
  volume = {195},
  pages = {703--724},
  issn = {0038-092X},
  doi = {10.1016/j.solener.2019.11.087},
  url = {http://www.sciencedirect.com/science/article/pii/S0038092X19311831},
  urldate = {2020-09-11},
  langid = {english}
}

@article{markowitzPortfolioSelection1952,
  title = {Portfolio Selection},
  author = {Markowitz, Harry},
  date = {1952},
  journaltitle = {The Journal of Finance},
  volume = {7},
  pages = {77--91},
  issn = {1540-6261},
  doi = {10.1111/j.1540-6261.1952.tb01525.x},
  url = {http://dx.doi.org/10.1111/j.1540-6261.1952.tb01525.x},
  number = {1}
}

@report{nelsonInvestigatingEconomicValue2018,
  title = {Investigating the Economic Value of Flexible Solar Power Plant Operation},
  author = {Nelson, Jimmy and Kasina, Saamrat and Stevens, John and Moore, Jack and {Arne Olson} and {Mahesh Morjaria} and Smolenski, John and Aponte, Jose},
  date = {2018-10},
  institution = {{Energy and Environmental Economics, Inc.}},
  location = {{San Francisco, USA}},
  url = {https://www.ethree.com/wp-content/uploads/2018/10/Investigating-the-Economic-Value-of-Flexible-Solar-Power-Plant-Operation.pdf},
  urldate = {2020-07-29}
}

@article{novacheckDiversifyingWindPower2017,
  title = {Diversifying Wind Power in Real Power Systems},
  author = {Novacheck, Joshua and Johnson, Jeremiah X.},
  date = {2017-06-01},
  journaltitle = {Renewable Energy},
  shortjournal = {Renewable Energy},
  volume = {106},
  pages = {177--185},
  issn = {0960-1481},
  doi = {10.1016/j.renene.2016.12.100},
  url = {http://www.sciencedirect.com/science/article/pii/S0960148116311697},
  urldate = {2020-10-01},
  langid = {english}
}

@online{onsCurvaCargaHoraria2019,
  title = {Curva de Carga Horária},
  author = {{ONS}},
  date = {2019},
  url = {http://www.ons.org.br/Paginas/resultados-da-operacao/historico-da-operacao/curva_carga_horaria.aspx},
  urldate = {2020-07-05},
  organization = {{ONS - Operador Nacional do Sistema Elétrico}}
}

@article{ramirezcamargoSimulationMultiannualTime2020,
  title = {Simulation of Multi-Annual Time Series of Solar Photovoltaic Power: {{Is}} the {{ERA5}}-Land Reanalysis the next Big Step?},
  shorttitle = {Simulation of Multi-Annual Time Series of Solar Photovoltaic Power},
  author = {Ramirez Camargo, Luis and Schmidt, Johannes},
  date = {2020-12-01},
  journaltitle = {Sustainable Energy Technologies and Assessments},
  shortjournal = {Sustainable Energy Technologies and Assessments},
  volume = {42},
  pages = {100829},
  issn = {2213-1388},
  doi = {10.1016/j.seta.2020.100829},
  url = {http://www.sciencedirect.com/science/article/pii/S221313882031256X},
  urldate = {2020-11-19},
  langid = {english}
}

@software{rcoreteamLanguageEnvironmentStatistical2020,
  title = {R: {{A}} Language and Environment for Statistical Computing},
  author = {{R Core Team}},
  date = {2020},
  location = {{Vienna, Austria}},
  url = {https://www.R-project.org/},
  organization = {{R Foundation for Statistical Computing}}
}

@article{rockafellarOptimizationConditionalValueatrisk2000,
  title = {Optimization of Conditional Value-at-Risk},
  author = {Rockafellar, R. Tyrrell and Uryasev, Stanislav},
  date = {2000},
  journaltitle = {The Journal of Risk},
  shortjournal = {JOR},
  volume = {2},
  pages = {21--41},
  issn = {14651211},
  doi = {10.21314/JOR.2000.038},
  url = {http://www.risk.net/journal-of-risk/technical-paper/2161159/optimization-conditional-value-risk},
  urldate = {2019-09-24},
  langid = {english},
  number = {3}
}

@article{rombautsOptimalPortfoliotheorybasedAllocation2011,
  title = {Optimal Portfolio-Theory-Based Allocation of Wind Power: {{Taking}} into Account Cross-Border Transmission-Capacity Constraints},
  shorttitle = {Optimal Portfolio-Theory-Based Allocation of Wind Power},
  author = {Rombauts, Yannick and Delarue, Erik and D’haeseleer, William},
  date = {2011-09-01},
  journaltitle = {Renewable Energy},
  shortjournal = {Renewable Energy},
  volume = {36},
  pages = {2374--2387},
  issn = {0960-1481},
  doi = {10.1016/j.renene.2011.02.010},
  url = {http://www.sciencedirect.com/science/article/pii/S0960148111000863},
  urldate = {2018-10-15},
  number = {9}
}

@article{roquesOptimalWindPower2010,
  title = {Optimal Wind Power Deployment in {{Europe}} -- {{A}} Portfolio Approach},
  author = {Roques, Fabien and Hiroux, Céline and Saguan, Marcelo},
  date = {2010-07},
  journaltitle = {Energy Policy},
  volume = {38},
  pages = {3245--3256},
  issn = {03014215},
  doi = {10.1016/j.enpol.2009.07.048},
  url = {http://linkinghub.elsevier.com/retrieve/pii/S030142150900545X},
  urldate = {2018-03-29},
  langid = {english},
  number = {7}
}

@article{santos-alamillosExploringMeanvariancePortfolio2017,
  title = {Exploring the Mean-Variance Portfolio Optimization Approach for Planning Wind Repowering Actions in {{Spain}}},
  author = {Santos-Alamillos, F. J. and Thomaidis, N. S. and Usaola-García, J. and Ruiz-Arias, J. A. and Pozo-Vázquez, D.},
  date = {2017-06-01},
  journaltitle = {Renewable Energy},
  shortjournal = {Renewable Energy},
  volume = {106},
  pages = {335--342},
  issn = {0960-1481},
  doi = {10.1016/j.renene.2017.01.041},
  url = {http://www.sciencedirect.com/science/article/pii/S0960148117300514},
  urldate = {2020-09-27},
  langid = {english}
}

@incollection{sarykalinValueatRiskVsConditional2008,
  title = {Value-at-{{Risk}} vs. {{Conditional Value}}-at-{{Risk}} in {{Risk Management}} and {{Optimization}}},
  booktitle = {State-of-the-{{Art Decision}}-{{Making Tools}} in the {{Information}}-{{Intensive Age}}},
  author = {Sarykalin, Sergey and Serraino, Gaia and Uryasev, Stan},
  editor = {Chen, Zhi-Long and Raghavan, S. and Gray, Paul and Greenberg, Harvey J.},
  date = {2008-09},
  pages = {270--294},
  publisher = {{INFORMS}},
  doi = {10.1287/educ.1080.0052},
  url = {http://pubsonline.informs.org/doi/abs/10.1287/educ.1080.0052},
  urldate = {2019-09-24},
  isbn = {978-1-877640-23-0},
  langid = {english}
}

@article{scalaPortfolioAnalysisGeographical2019,
  title = {Portfolio Analysis and Geographical Allocation of Renewable Sources: {{A}} Stochastic Approach},
  shorttitle = {Portfolio Analysis and Geographical Allocation of Renewable Sources},
  author = {Scala, Antonio and Facchini, Angelo and Perna, Umberto and Basosi, Riccardo},
  date = {2019-02-01},
  journaltitle = {Energy Policy},
  shortjournal = {Energy Policy},
  volume = {125},
  pages = {154--159},
  issn = {0301-4215},
  doi = {10.1016/j.enpol.2018.10.034},
  url = {http://www.sciencedirect.com/science/article/pii/S0301421518306888},
  urldate = {2020-09-25},
  langid = {english}
}

@article{schmidtOptimalMixSolar2016,
  title = {An Optimal Mix of Solar {{PV}}, Wind and Hydro Power for a Low-Carbon Electricity Supply in {{Brazil}}},
  author = {Schmidt, Johannes and Cancella, Rafael and Pereira, Amaro O.},
  date = {2016-01-01},
  journaltitle = {Renewable Energy},
  shortjournal = {Renewable Energy},
  volume = {85},
  pages = {137--147},
  issn = {0960-1481},
  doi = {10.1016/j.renene.2015.06.010},
  url = {http://www.sciencedirect.com/science/article/pii/S0960148115300331},
  urldate = {2020-10-24},
  langid = {english}
}

@article{shahriariCapacityValueOptimal2018,
  title = {The Capacity Value of Optimal Wind and Solar Portfolios},
  author = {Shahriari, Mehdi and Blumsack, Seth},
  date = {2018-04-01},
  journaltitle = {Energy},
  shortjournal = {Energy},
  volume = {148},
  pages = {992--1005},
  issn = {0360-5442},
  doi = {10.1016/j.energy.2017.12.121},
  url = {http://www.sciencedirect.com/science/article/pii/S0360544217321643},
  urldate = {2020-09-26},
  langid = {english}
}

@article{thomaidisOptimalManagementWind2016a,
  title = {Optimal Management of Wind and Solar Energy Resources},
  author = {Thomaidis, Nikolaos S. and Santos-Alamillos, Francisco J. and Pozo-Vázquez, David and Usaola-García, Julio},
  date = {2016-02-01},
  journaltitle = {Computers \& Operations Research},
  shortjournal = {Computers \& Operations Research},
  volume = {66},
  pages = {284--291},
  issn = {0305-0548},
  doi = {10.1016/j.cor.2015.02.016},
  url = {http://www.sciencedirect.com/science/article/pii/S0305054815000556},
  urldate = {2020-10-01},
  langid = {english}
}

@preamble{ "\ifdefined\DeclarePrefChars\DeclarePrefChars{'’-}\else\fi " }

\hypertarget{appendix-appendix}{%
\appendix}

\numberwithin{equation}{section}
\hypertarget{proof}{%
\section{Proof}\label{proof}}

We compare two different formulations to achieve efficient frontiers. Both are based on a two objective optimization with a fixed parameter. The fixed parameter is installed capacity in the first model (Model C) and expected generation in the second one (Model G). A portfolio is non-dominated, i.e. belongs to the efficient frontier, if there is no other portfolio which is better in one of the parameters (standard deviation and capacity factor) while it is at least equal in the other parameter.

\begin{definition}
The set of all efficient portfolios is the efficient frontier.
\end{definition}
\begin{definition}
Let $\mathcal{E}_{C}$ be the efficient frontier obtained when solving Model C, i.e., the model formulation that maximizes portfolio generation and minimizes standard deviation at a fixed capacity. 
\end{definition}
\begin{definition}
Let $\mathcal{E}_{G}$ be the efficient frontier obtained when solving Model G, i.e., the model formulation that minimizes portfolio installed capacity and minimizes standard deviation at a fixed generation.
\end{definition}
\begin{definition}
Let $p^{C} \in \mathcal{E}_{C}$ be a portfolio from Model C.
\end{definition}
\begin{definition}
Let $p^{G} \in \mathcal{E}_{G}$ be a portfolio from Model G.
\end{definition}
\begin{definition} \label{DefCVpoint}
Let $cv \in \mathcal{E}_{C}$ be the portfolio, in Model C, with lowest ratio between standard deviation and capacity factor ($\frac{\sigma_{cv^{C}}}{\mu_{cv}}$).
\end{definition}
\begin{definition}
Given a pair of portfolios $p_x$ and $p_y$, they are equivalent ($p_x \equiv p_y$) if the share of each plant that compose the portfolios is equal in both portfolios.
\end{definition}

\begin{theorem} \label{th:minCV}
Let $cv^G \in \mathcal{E}_{G}$, $cv^C \in \mathcal{E}_C$ and $cv^G \equiv cv^C$, then:
\begin{equation}
\sigma_{cv^G} < \sigma_{p} \quad \forall (p \neq cv) \in \mathcal{E}_{G}
\end{equation}
\end{theorem}

\begin{proof}
Let $p \in \mathcal{P}$ be any possible portfolio, with expected generation $G_p$, installed capacity $C_p$ and standard deviation $\sigma_p$. In order to verify if $p$ is equivalent to a portfolio in the efficient frontier $\mathcal{E}_{C}$, it is necessary to obtain the standard deviation of the portfolio equivalent to $p$ ($\sigma_{p^C}$) that respects the fixed installed capacity constraint ($C_{F}$) used to define $\mathcal{E}_{C}$. This is equivalent to normalizing the value by the portfolio installed capacity.

Therefore:

\begin{equation}
\sigma_{p^C} = \frac{\sigma_{p}}{C_p} C_{F} \label{eq:DefSDc}
\end{equation}

Conversely, \(\sigma_{p^G}\) is the standard deviation of the portfolio equivalent to \(p\) normalized to the value of fixed generation ($G_F$) used to define $\mathcal{E}_{G}$ .

\begin{equation}
\sigma_{p^G} = \frac{\sigma_p}{G_p} G_F \label{eq:DefSDG}
\end{equation} 

Therefore:

\begin{equation}
\sigma_{p^G} = \sigma_{p^C} \frac{C_p}{G_p} \frac{G_F}{C_F} \label{eq:SDgToSDc}
\end{equation}

\begin{equation}
\sigma_{p^C} = {\sigma_{p^G}}\frac{G_p}{C_p} \frac{C_F}{G_F} \label{eq:SDcToSDg}
\end{equation}

As \(\frac{G_p}{C_p}\) is the portfolio capacity factor (\(\mu_p\)), equation \eqref{eq:SDcToSDg} can be rewritten as:

\begin{equation}
\sigma_{p^C} = \sigma_{p^G} \mu_p \frac{G_F}{C_F} \label{eq:SDcToSDg2}
\end{equation}

By Definition \ref{DefCVpoint}, we have

\begin{equation}
\frac{\sigma_{cv^{C}}}{\mu_{cv}} < \frac{\sigma_{p^C}}{\mu_p}. \label{eq:MinCV}
\end{equation}

By application of Equation (\ref{eq:SDcToSDg}) we get
\begin{equation}
\frac{\sigma_{cv^G} \mu_{cv} \frac{G_F}{C_F}}{\mu_{cv}} < \frac{\sigma_{p^G} \mu_p \frac{G_F}{C_F}}{\mu_p} \label{eq:MinCV2}
\end{equation}
or equivalently
\begin{equation}
\sigma_{cv^G} < \sigma_{p^G} \label{eq:MinPoint}.
\end{equation}

\end{proof}

\begin{theorem}
Any portfolio in the efficient frontier of Model C whose standard deviation is lower than the standard deviation of $cv$ portfolio does not have an equivalent portfolio in Model G's efficient frontier. 
\begin{equation}
\sigma_{p^C} < \sigma_{cv^C} \implies \nexists  (p^{G} \in \mathcal{E}_{G}) \equiv p^{C}  \quad \forall {p^C} \in \mathcal{E}_{C}
\end{equation}
\end{theorem}

\begin{proof}
In order to show that a portfolio $p$ does not belong to the efficient frontier, it is sufficient to show that exists, in the efficient frontier, any portfolio that surpasses $p$ in any parameter while it is at least equal in the other parameter. According to Theorem \ref{th:minCV}, $\sigma^G_{cv} < \sigma^G_{p}$. Therefore, to prove the current theorem, it is sufficient to show that $\mu_{cv} \ge \mu_{p}$.
Therefore, by Definition \ref{DefCVpoint}:
\begin{equation}
\frac{\sigma_{cv^{C}}}{\mu_{cv}} < \frac{\sigma_{p^C}}{\mu_p} \label{eq:MinCV3}
\end{equation}

If $\sigma_{p^C} < \sigma_{cv^C}$, it implies that

\begin{equation}
\frac{\mu_{p}}{\mu_{cv}} < \frac{\sigma_{p^C}}{\sigma_{cv^{C}}} < 1, \label{eq:RatioMuSigma}
\end{equation}
which simplifies to
\begin{equation}
\mu_{p} < \mu_{cv} \label{eq:lowMu}.
\end{equation}
\end{proof}

\section{Conditional Value-at-Risk}
\label{AppendixCvar}

Modern Portfolio Theory's origins are firmly rooted in finance, while we study security of supply in electricity grids. The objective of this paper is essentially a reinterpretation of MPT's finance roots towards an energy economics perspective. Each problem domain is equipped with a distinct notion of risk. In finance, risk is understood as a financial loss resulting from a investment position. In contrast within this work, we reframe risk as a loss of load in an electricity grid.
In order to do to this we require a sound understanding of the interplay of MPT's key concepts: the loss function \& the risk measure. We review them briefly in the following.

\subsection{What is a loss function?}

Modern Portfolio Theory is a tool for risk management.
As such it aims to avoid undesirable events.
The exact nature of these events may vary and needs to be precisely defined by the modeller in terms of a so-called loss-function $f(Y,\lambda)$.
For an intuitive understanding of $f$ it is particularly necessary to understand the difference between the two inputs $Y$ and $\lambda$.
Both inputs should be understood as vectors.
$Y$'s components are random variables, while $\lambda$ are non-random parameters that acts as a decision variable in an optimization problem that involves $Y$. The interpretation here is that $Y$ models the risks we do not control, while $\lambda$ models what we do control, namely our strategy to deal with the risk of $Y$.
Hence, a loss depends as much on chance (as governed by $Y$) as on our risk-management strategy (as given by $\lambda$).
Consequently, fixing the parameter $\lambda$ means fixing the risk-management strategy. Hence, $f(\bullet,\lambda)$ is purely a random variable ("What variations in loss are expected under a fixed risk-management strategy"), $f(Y_i,\bullet)$ is purely an optimization problem ("What is the best risk management strategy against a particular outcome of the random variable?") and $f(\bullet,\bullet)$ is both varying randomly and an optimization problem ("What is the best risk management strategy against a random variable and all its possible outcomes?").

In the context of portfolio management, $Y_i$ would be the values of assets, while $\lambda_i$ would be the ratio of assets in a portfolio that can be used to hedge against variations in value.
Moreover, a loss may be modelled non-uniquely by various functional forms. 
In the context of security of electricity supply, we reinterpret $Y_i$ as the production of plant $i$ and the special case of $-Y_L$ as the electricity demand. Furthermore, to denote the difference between the standard approach and our proposal, we relabel the decision variables $\lambda_i=P_i$ that denote installed capacity of each plant in a portfolio that can be used to hedge against electricity production below demand. We close by remarking that $P_L$ is a special case. It should usually not be modelled as a decision variable, since grid suppliers typically desire to avoid demand cuts.

\subsection{What is the CVaR?}

The Conditional Value-at-Risk (CVaR) in-turn is a method to quantify the risk incurred by a loss-function $f(Y,\lambda)$.
Hence, the CVaR depends directly on the employed risk function. 
Moreover, since a loss may be modelled non-uniquely by various functional forms $f$, different risk models may lead to differing CVaRs.
CVaR is the mean of the losses that have cumulative probability over a certain safety threshold $\beta$. 
The CVaR operates under the assumption that a certain safety threshold VaR (Value-at-Risk) should not be exceeded and that conversely, the losses below the threshold are negligible. Once such a threshold has been set the \textit{CVaR is the mean loss exceeding the safety limit VaR.} 

\subsection{How to compute the CVaR?}

The CVaR has several appealing mathematical properties, but it can not be computed directly. However, this problem can be side-stepped.
In a lengthy technical derivation, \textcite{rockafellarOptimizationConditionalValueatrisk2000} show that the auxilliary function
\begin{equation}
    F_{\beta}(\lambda,\alpha) = \alpha + \frac{\mathbb{E}\left( \max(f(Y,\lambda)-\alpha),0 \right)}{(1-\beta)}
    \label{eq:aux_f}
\end{equation}
can be used to compute the CVaR indirectly.
Here, $f(Y,\lambda)$ is the loss-function with $Y$ being the random variable characterising the risk and $\lambda$ the decision variables of the risk-management strategy (i.e. portfolio weights). Additionally, we encounter two new model parameters $\alpha$ and $\beta$. 

$\beta$ denotes the desired level of security (i.e. the probability not to exceed the security threshold $\alpha$). Hence, $\beta$ is a probability between 0 and 1. Note, that while $\beta$ denotes the probability to remain within the safety limit, $(1-\beta)$ denotes the probability to make loss (that exceeds $\alpha$). Therefore, in \eqref{eq:aux_f} the expected loss beyond $\alpha$ is divided by the probability of the severe loss occurring.

Note, that the evaluation of the expectation values is not necessarily easy, but depends on the properties of the loss $f$.
However, a possible approach is to simply estimate the expectation by sampling $Y_{t}$ from the random variable $Y$
\begin{equation}
    \alpha + \frac{\mathbb{E}\left( \max(f(Y,\lambda)-\alpha),0 \right)}{(1-\beta)}
    \approx
    \alpha + \frac{1}{(1-\beta)}\frac{\sum_{t=1}^{T}\max(f(Y_{t},\lambda)-\alpha,0)}{T}
    \approx
    \alpha + \frac{1}{(1-\beta)}\frac{\sum_{m=1}^{M}\max(f(Y_{t_{m}},\lambda)-\alpha,0)}{M},
\end{equation}
where the first approximation will hold if the sampling correctly represents $Y$ and the second approximation is a subsampling  with $\forall m ~ in ~ M \subseteq T$ that will hold if the number of samples is not too small or represents the complete time series well.

The way to minimize CVaR is to compute
\begin{equation}
    \min_{\alpha,\lambda} F_{\beta}(\lambda,\alpha) = \hat{F}_{\beta}(\hat{\lambda},\hat{\alpha}).
    \label{eq:aux_f2}
\end{equation}
Indeed, minimizing \eqref{eq:aux_f} yields the value of the CVaR and the arguments $\hat{\lambda},\hat{\alpha}$ yield additional information, where  $\hat{\alpha}$ equals the so called Value-at-Risk and $\hat{\lambda}$ the optimal decision variables. The minimization of \eqref{eq:aux_f2} is efficiently possible, if $F_{\beta}$ is convex. It can be shown that $F_{\beta}$ is convex, if the loss $f$ and $\alpha$ domain are convex too. This makes the CVaR calculable via assumptions on $f$. 
We can simplify
\begin{align}
    \min_{\alpha,\lambda} \alpha + \frac{1}{(1-\beta)}\frac{\sum_{m=1}^{M}\max(f(Y_{t_{m}},\lambda)-\alpha,0)}{M} 
    = &\min_{\alpha,\lambda, Z_{t_{m}}} &\alpha + \frac{1}{(1-\beta)}\frac{\sum_{m=1}^{M}Z_{t_{m}}}{M} ~ & \label{eq:MinCVaR} \\
    & ~ s.t.:  &Z_{t_{m}} \geq f(Y_{t_{m}},\lambda)-\alpha ~ & \forall m \in M \\
    &  &Z_{t_{m}} \geq 0 ~ & \forall m \in M
\end{align}
by rewriting the maximum into a positivity constraint. 

Instead of minimizing CVaR value, it is possible to constraint it. In this case, it is necessary to define the upper limit to CVaR. Therefore, the objective function \eqref{eq:MinCVaR} can be turned into the constraint

\begin{equation}
\alpha + \frac{1}{(1-\beta)}\frac{\sum_{m=1}^{M}Z_{t_{m}}}{M} \le \omega 
\end{equation}

Here, $\omega$ can be seen as the acceptable limit of losses. $\omega$ will be  a number of the same unit as the losses. This is exemplified for instance by either $\omega \longrightarrow \infty$ where \eqref{eq:aux_f} reduces to 0 ("If everything is acceptable there is no risk") or setting $\omega=0$ where \eqref{eq:aux_f} reduces to the expectation of the loss ("If nothing is acceptable we risk as much as we expect to lose").

We have remained until now completely general without specifying the form of the loss $f$. Therefore the derivation is  actually independent of the problem domain. However, for explicitness sake we finally choose the loss function such that it models loss of load. We set  the decision variables $\lambda_i=P_i$ and rewrite
\begin{equation}
    f(Y_{t_{m}},\lambda) = \sum_{i=1}^{N}(Y_{t_{m},i} P_i - Y_{t_{m},L} P_L),
\end{equation}
where $Y$ ratio of plant generation per installed capacity and $P$ the installed capacity of plants in the portfolio. 
We decompose the normalized generation into (possibly sub-sampled) normalized generation per time $Y=\sum_{t=1}^{T}Y_{t} \approx \sum_{t_{m}=1}^{M} Y_{t_{m}}$ and each time steps is decomposed into positive plantwise generation contributions $\sum_{i=1}^{N}Y_{t_{m},i}$ of $N$ plants and the negative system demand $- Y_{t_{m},L} P_L)$, resulting in the energy balance. Note, that in the only slight deviation from the standard case of portfolio theory $P_{L}$ is not a decision variable, since it models the demand that is not under our control. With this choice, we have derived the risk constraints in Equations \eqref{eq:RiskCVaR}, \eqref{eq:RiskZ} and \eqref{eq:RiskNonNeg}, with the only difference that these constraints refer to the energy balance -- in which we define a lower limit to CVaR -- instead of the loss function -- in which we define an upper limit. Therefore, some adaptations were made to the constraints to reflect this difference.
\end{document}